%% file: TDTv4.tex
\documentclass[12pt,draftcls,peerreview, onecolumn]{IEEEtran}
\IEEEoverridecommandlockouts  
\usepackage{times}
\input{macros-ieee.tex}
\usepackage{epsfig}
\usepackage{cite,url}
\usepackage{graphicx,subfigure,color}
\usepackage{amsmath,amsfonts,amssymb,mathrsfs,wasysym,stackrel}
\usepackage{float}
\usepackage{xspace}
\usepackage{multirow,tabularx}
\usepackage{booktabs}
\usepackage{tikz}

\newcommand{\remove}[1]{}

\title{ Role of  Externally Provided Randomness in Stochastic Teams and Zero-sum Team Games}

\author{ Rahul Meshram
	\thanks{Rahul Meshram is  with  the Electronics and Communication Engineering Department, 
 Indian Institute of Information Technology Allahabad.  
Email-\texttt{rhmeshram@iiita.ac.in} 
   } 
}
    
\begin{document}

\maketitle

\thispagestyle{empty}
\pagestyle{empty}

\begin{abstract}
Stochastic team decision problem is extensively studied in literature and the existence of optimal solution is obtained in recent literature. The value of information in statistical problem and decision theory is classical problem. Much of earlier does not qualitatively describe role of externally provided private and common randomness in stochastic team problem and team vs team zero sum game. 

In this paper, we study the role of extrenally provided private or common randomness in stochastic team decision. We make observation that the randomness independent of environment does not benefit either team but randomness dependent on environment benefit teams and decreases the expected cost function. We also studied LQG team game with special information structure on private or common randomness. 
We extend these study to problem team vs team zero sum game. We show that if a game admits saddle point solution, then private or common randomness independent of environment does not benefit either team. We also analyze the scenario when a team with having more information than other team which is dependent on environment and game has saddle point solution, then team with more information benefits. This is also illustrated numerically for LQG team vs team zero sum game. Finally, we show for discrete team vs team zero sum game that private randomness independent of environment benefits team when there is no saddle point condition. Role of common randomness is discussed for discrete game.      
\end{abstract}



\section{Introduction}
\label{sec:intro}
A team decision problem consists of two or more of 
decision makers (DMs) or players that make decisions in a random environment where the \textit{information} of each DM is a (possibly partial) observation about the random environment. A DM takes an action as a function of the information; this function is referred to the as the \textit{decision rule}. DMs choose the decision rule to jointly minimize an expected \textit{cost}. 

If the decision makers had identical observations, then the multiple decision makers could be clumped together as a single decision maker and the problem reduces to that of a stochastic optimization or control problem. Of interest to us here is the case where information is asymmetric, whereby there is no obvious method of aggregating it. A team decision problem is in essence a \textit{decentralized} stochastic optimal control problem. Problems with structure appear in a variety of settings for example in sensor networks. The decision makers could be sensors situated at different locations. These sensors observe the environment through different, possibly imperfect, channels and under this information structure, the sensors have to act collectively to minimize a certain cost function. 

In this paper we consider the role of \textit{externally provided private and common randomness} in stochastic teams and in stochastic team v/s team zero-sum games. In the setup described above, it is conceivable that an external source provides randomness to the players.  
This randomness may or may not be correlated with their observations, and it may or may not be correlated across players. This randomness increases the set of achievable joint distributions on the joint action space of the DMs. Our goal is to understand the role of such randomness in a team problem and a team v/s team game.

Qualitatively, there are three of kinds randomness that an external source may provide. First, the source of randomness could be a \textit{coordinator} -- namely an entity that \textit{mixes} the actions of the DM by randomizing.
Mathematically, this randomness is independent of the observations of the DMs, but it may be correlated across DMs. This correlation makes this randomness distinct from the usual notion of ``randomized policies'', in which the randomization is independently performed by each DM. The second kind of randomness, may be imagined as a \textit{counsellor} -- this entity accesses the observations of each decision maker and provides a common message to all DMs. This kind of randomness is correlated with the environment. 
The sources of randomness mentioned above are relevant for team problems as well as team v/s team games. The third kind, which is relevant only in the team v/s team game, is that of a \textit{mole} or a \textit{spy}. This source of randomness provides information about observations of the opposite team.

Our interest is in qualitatively understanding the role of the kinds of randomness mentioned above and quantifying it. 
We make following contribution in  this work.
\begin{enumerate}
\item A team decision problem:
\begin{itemize}
\item We show that if a coordinator provides private or common randomness independent of the environment, then it cannot improve the cost. 
\item We show that common randomness dependent on the environment can improve the team cost. For a certain class of LQG team problems, we show that if the information of each player is replaced by a convex combination of the information of all players, then the team improves its cost. 
\end{itemize}
\item Team v/s team zero-sum game:
\begin{itemize}
\item We show randomness independent the an environment does not benefit teams if the zero-sum game admits saddle point solution. 
\item We prove that a team having more information than other team, benefit and decreases the cost function for minimizing team when randomness is dependent on environment.
\item For LQG team zero-sum games we illustrate that  common randomness dependent on the environment leads to an improvement in the optimal team cost. 
\item Finally, We give an example of a discrete team v/s team zero sum game, without a pure  strategy saddle point, we also show that private and common randomness independent of an environment benefit teams. But it may not have Nash equilibrium.

\end{itemize}
\end{enumerate}

\subsection{Related Work}

Early work on team decision problem in aspect of 
an organization theory studied in \cite{Marschak55}; where author  used the concepts from game theory and statistical decision theory. A general formulation of team decision problem  are described in \cite{Radner62} and person by person optimality condition is established to solve the distributed team decision problem.
Furthermore, the team decision problem extended to a LQG team  problem 
 in \cite{Ho78,Ho80}.  They investigated  static and dynamic  LQG team decision problem and explored its connection with information theory and economics. 
 In LQG team problem, there is a unique optimal solution, linear in information and it is obtained via solving person by person optimal condition. They also studied dynamic LQG team with partial-nested information structure. Moreover, the symmetric static  team problem studied in \cite{Schoute78} and have shown that the optimal strategy for a symmetric team problem not necessarily a pure strategy but it can have randomized strategy.


Two-team zero-sum game in LQG problem studied in \cite{Ho74}, and they show that  team having extra information not necessary  ameliorate the expected  loss.  
Apart from a team decision problem, the role of common randomness in multi-agent distributed control problem is analyzed in \cite{Anantharam07}. 
Our work is inspired from \cite{Ho78,Ho74,Anantharam07}. Role of common randomness is not quantified in \cite{Ho78,Ho74}, whereas we discuss  role of the private and common randomness in team decision problem.

The value of information for  statistical  problems is first introduced in  \cite{Blackwell51,Blackwell53}. This is further extended to decision problems in \cite{Marschak68}, and author have shown that increasing information lead to increasing in utility. 
Early work on role of increasing information in two person game problem is presented in \cite{Ho73}.The surprising finding is presented in \cite{Ho73}, where author finds that increasing informativeness leads to decreasing performance.   The value of information available to players with two-person zero sum game is studied in \cite{Witsenhausen71}. As the additional information increased for a player, may lead to solution toward ideal optimality condition  when there is a saddle point condition exists. This result further motivated study on value of information in  team vs team zero sum game and similar result have shown for LQG team vs team zero sum game. 
In \cite{Basar74}, the value of information for two players non zero sum game is developed, and they have show that in LQG model with better informed player, it decreases the average Nash cost for both players but in duopoly problem, the better informed player benefits only.  

The great reference for stochastic team decision problem is \cite{Yuksel13}. In this book, authors discussed fundamental of team decision problem, sequential team decision problem, comparison of information structure, topological properties of information structure and its application to communication channel. It has motivated further research on team decision problem in recent time. 

There are flurry of research activities on static team problems and their existence of solution.  In \cite{Gupta15}, the class of sequential team problem is studied with a certain information structure and existence of optimal strategies are proved.  Further, they have shown the existence of optimal solution for team problem under weaker assumptions, i.e., assumption  on cost function to be bounded and continuous, action space of agent to be compact or not compact and observation  satisfies technical condition. The ideas from weak convergence in probability theory is used to show convergence of measure of joint probability of actions. In \cite{Saldi19}, author extended study of \cite{Gupta15}, further weaken assumptions their. They have shown the existence of optimal strategies for static teams and topology on set of policies are introduced. In \cite{Yuksel17}, authors studied convexity properties of strategy spaces and discussed redundancy of common or private information that is independent of randomness for static team. Though this result is similar to ours, their proof differs from our method.   The role of common information in dynamic stochastic game is studied in \cite{Gupta14}, where asymmetric of common information is considered among players.  In \cite{Gupta20}, the existence of optimal solution to static team problem under private and common information structure is developed using topology of information and space of measures. 
Early ideas developed in \cite{Blackwell51,Blackwell53,Basar74,Ho73,Marschak68} on role of information are derived for zero sum game under slightly weaker assumptions in \cite{Hogeboom-Burr20} and have shown existence of saddle point equilibrium.

But these paper do not provide qualitative comparison of role of externally provided private and common randomness in static team and team vs team zero sum game.
%
%
%
%
%
%




The rest of the paper is organized as follows. Private and common randomness in
static team decision 
problem described 
 in Section~\ref{sec:pricomrandstaticteam}. Role of private and common randomness
in static team vs team zero-sum game  developed in
 Section~\ref{sec:teamvsteamprivatecom}. Finally, concluding remarks and future direction of research presented in Section~\ref{sec:discussion-concluson}.

\def\DM{{\rm DM}}
\def\Nscr{{\cal N}}
\def\Real{{\mathbb{R}}}
\def\TO{{\rm TO}}
\def\PBP{{\rm PBP}}
\def\mi{^{-i}}
\def\Tr{{\rm Tr}}
\section{Private and common randomness in  static team problem}
\label{sec:pricomrandstaticteam}
\def\Jscr{\mathcal{J}}
\subsection{Team decision problem}
\label{subsec:team-prob}
Consider a team decision problem having $ N $ decision makers 
$\DM_1,\hdots,\DM_N$ in a team and let $ \Nscr=\{1,\hdots,N\} $. Let $\xi$ be a random vector taking values in a space $ \Xi $ denoting the state of nature or 
an environment; let its distribution be 
$\mathbf{P}(\cdot)$.
Define $y_i := \eta_i(\xi)$ for a measurable function $ \eta_i $ to be 
 the information observed by $\DM_i$ and let $Y_i  $ be the space of $ y_i. $
Let $U_i \subseteq \Real^{m_i}$, $ m_i \in \Nbb $ denote the set of actions of $\DM_i$. The strategy space of $ \DM_i $ is $\Gamma_i$, the space of measurable functions
$ \gamma_i $ mapping $ Y_i $ to $ U_i $ and 
an action $ u_i $ is given by $u_i = \gamma_i(y_i)$. Without loss of generality we take $ U_i \subseteq \mathbb{R} $ for all $ i $, since a DM with a $ \Real^{m_i} $-valued strategy can be considered as $ m_i $ separate DMs with $ \Real $-valued strategies; thus $ m_i=1 $ for all $ i\in \Nscr. $ 
Let 
\begin{align*}
u &:= (u_1,\hdots,u_N), \quad \gamma:=(\gamma_1,\hdots,\gamma_N), \\
 u\mi &:= (u_1,\hdots,u_{i-1},u_{i+1},\hdots,u_N), \quad \gamma\mi := (\gamma_1,\hdots,\gamma_{i-1},\gamma_{i+1},\hdots,\gamma_N) 
\end{align*}
The cost function is measurable function $ \kappa : U \times \Xi \rightarrow \Real$, where $ U: \prod_{i \in \Nscr} U_i $ and let \[ \Jscr(\gamma) = \mathbb{E}_{\xi}[\kappa(u_1 = \gamma_1(\eta_1(\xi)),\hdots,u_N=\gamma_N(\eta_N(\xi)),\xi)].\]

A \textit{team optimal} solution of the above problem is defined as $ \gamma^* \in \Gamma := \prod_{i\in \Nscr} \Gamma_i $ such that  
\begin{equation}
\mathcal{J}_{{\rm TO}}^{*} \triangleq \Jscr(\gamma^*)= \min_{\gamma \in \Gamma}\mathcal{J}(\gamma) = \min_{\gamma
 \in \Gamma} \mathbb{E}_{\xi}[\kappa(u_1 = \gamma_1(\eta_1(\xi)),\hdots,u_N=\gamma_N(\eta_N(\xi)),\xi)]. \label{eq:teamopt}
\end{equation}
We assume throughout that a team optimal solution exists and use `min' instead of `inf'.  
A related concept, called the \textit{person by person optimal} solution is a $ \gamma \in \Gamma $ such that
\[\Jscr_{\rm PBP}^*= \Jscr(\gamma)= \min_{\gamma_i' \in \Gamma_i}\mathcal{J}(\gamma_i' ,{\gamma}^{-i} )  \qquad \forall \ i \in \Nscr.\]

\subsection{Externally provided randomness}
\def\Pbb{\mathbb{P}}
We now introduce externally provided randomness, beginning with private randomness. Suppose $ \DM_i $ chooses $ u_i $ randomly from $ U_i $ and let $ Q $ be the joint distribution of 
all variables involved, namely, $ \xi, y, u $. We say that the DMs have externally provided \textit{private randomness}, if 
\begin{equation}
 Q(u|y)  = \prod_{i \in \Nscr} Q(u_i|y_i). \label{eq:private2}
\end{equation}
This specification corresponds to the standard notion of 
\textit{randomized policies} in stochastic control or \textit{behavioral strategies} in stochastic games, wherein the action is chosen to be a random function of the information. 

In general one has
\[ Q(\xi,y,u) = Q(u|\xi,y)Q(\xi,y), \]
where $ Q(\xi,y,u) $ is the joint distribution of $ \xi,y,u $, 
$ Q(u|\xi,y) $ the conditional distribution of $ u $ given $\xi, y  $, and $ Q(\xi,y) $ is the marginal of $ \xi,y$
(evaluated at `$ \xi=\xi,y=y,u=u $'). When the randomness provided to DMs is independent of $ \xi $, we have 
\[ u|y \amalg \xi,\]
i.e, given $ y $ the choice of $ u $ is independent of $ \xi. $ Furthermore, the joint distribution of $ (\xi,y) $ is known; denote this distribution by $ P(\xi,y) $. 
Consequently, any joint distribution of $ \xi,y,u $ is given by
\begin{equation}
 Q(\xi,y,u) = Q(u|y)P(\xi,y). \label{eq:private1}
\end{equation}

To describe externally provided \textit{common} randomness, let  $w = (w_1,\hdots, w_N)$ be a random vector, $ w \amalg \xi $,  and assume that $w_i$ is externally provided to $\DM_i$ by a coordinator. With the additional information of $ w_i $, the  strategies $\gamma_i$ of $\DM_i$ are deterministic $ y_i \times w_i \rightarrow u_i$ mappings and $ \Gamma_i $ is the space of such strategies. For a given random vector $ w $ with distribution $ \Pbb $, the team optimal solution is defined analogously to \eqref{eq:teamopt}, as follows:
\begin{equation}
\min_{\gamma \in \Gamma}\mathcal{J}(\gamma) = \min_{\gamma
 \in \Gamma} \mathbb{E}_{\xi,w}[\kappa(u_1 = \gamma_1(\eta_1(\xi),w_1),\hdots,u_N=\gamma_N(\eta_N(\xi),w_N),\xi)]. \label{eq:teamoptw}
\end{equation}
Since $ \xi $ is independent of $ w, $ the expectation with respect to $ (\xi,w) $ is well defined once the marginals of $ \xi,w $ are defined.

\def\Qscr{{\cal Q}}
\def\Pscr{{\cal P}}
\subsection{Randomness independent of $\xi$}
In this section, we study the case of externally provided randomness that is independent of the state of nature  $ \xi $.
Our main result is that in a team problem, such randomness provides no benefit to the team. One may interpret this to mean that a team a gains nothing by hiring a coordinator whose sole role is that of mixing the actions of the team members without the use of any knowledge of the underlying state of nature or of the observations made by team members.

Let $ \Pscr(\cdots) $ be the set of joint distributions of on the space `$ \cdots $'. Let $ \Qscr$ be the set of joint distributions of random variables $ \xi,y,u $ that admit the decomposition above. i.e.,
\[ \Qscr = \{Q \in \Pscr(\Xi \times Y \times U)\ |\ Q\ {\rm satisfies\ } \eqref{eq:private1}, \eqref{eq:private2}\}.\]
Consider the following problem:
\begin{equation}
\Jscr^*_{\TO_P} \triangleq \min_{Q \in \Qscr}\mathbb{E}_{Q}[\kappa(u,\xi)]. \label{eq:teamoptq}
\end{equation}
From the decomposition of $ Q $ provided by \eqref{eq:private1}-\eqref{eq:private2}, it follows that \eqref{eq:teamoptq} is a multilinear program with separable constraints. Classical results show that \eqref{eq:teamoptq} admits a solution that is an extreme point, namely, one where $ u_i $ is a deterministic function of $ y_i. $ Consequently, $ \Jscr^*_{\TO_P} = \Jscr^*_{\TO} $ and we have the following result.
\begin{proposition}
In a static stochastic team problem, externally provided private randomness that is independent of the state of nature cannot improve the team's cost.
\end{proposition}
Proof is along the lines of proof of Proposition~\ref{lemma:staticteamproblem1}. We skip the proof details. 
 
Consider the following cost:
\[
\mathcal{J}_{\TO_C}^{*} = \min_{\gamma
 \in \Gamma, \mathbb{P}} \mathbb{E}_{\xi,w}[\kappa(\gamma_1(\eta_1(\xi),  w_1),\hdots,\gamma_N(\eta_N(\xi),  w_N),\xi)].
\]
This is the lowest cost that can be attained via common randomness.  The common randomness is independent of environment $\xi.$
\begin{proposition}
\begin{eqnarray}
\mathcal{J}_{\rm TO}^{*} = \mathcal{J}_{\TO_P}^{*} = \mathcal{J}_{\TO_C}^{*}.
\end{eqnarray}
\label{lemma:staticteamproblem1}
\end{proposition}
\begin{proof} 
It is enough to show that $\mathcal{J}_{\rm TO}^{*} = \mathcal{J}_{\TO_C}^{*}.$
Now consider
\[
\mathcal{J}_{\TO_C}^{*} = \min_{\gamma
 \in \Gamma, \mathbb{P}} \mathbb{E}_{\xi,w}[\kappa(\gamma_1(\eta_1(\xi),  w_1),\hdots,\gamma_N(\eta_N(\xi),  w_N),\xi)].
\]
Assuming $\{y_1, \cdots, y_N\}$  are well defined. Rewriting above expression, we obtain
\[
\mathcal{J}_{\TO_C}^{*} = \min_{\gamma
 \in \Gamma, \mathbb{P}} \mathbb{E}_{w} \mathbb{E}_{\xi/w}[\kappa(\gamma_1(\eta_1(\xi),  w_1),\hdots,\gamma_N(\eta_N(\xi),  w_N),\xi)].
\]
Since common randomness $w$ independent of $\xi,$  we have
\[
\mathcal{J}_{\TO_C}^{*} = \min_{\gamma
 \in \Gamma, \mathbb{P}} \mathbb{E}_{w} \mathbb{E}_{\xi}[\kappa(\gamma_1(\eta_1(\xi),  w_1),\hdots,\gamma_N(\eta_N(\xi),  w_N),\xi)].
\]
Now,  we split the minimization $ \min_{\gamma_1, \gamma_2 \in \Gamma, \mathbb{P}(w)} = \min_{\mathbb{P}(w)} \min_{\gamma_1, \gamma_2 \in \Gamma}$, we can also interchange $\min_{\gamma_1, \gamma_2 \in \Gamma} $ and expectation $\mathbb{E}_{w} $  since DMs can cooperate and communicate in team problem.
\[
\mathcal{J}_{\TO_C}^{*} = \min_{ \mathbb{P}}  \mathbb{E}_{w} \min_{\gamma  \in \Gamma}  \mathbb{E}_{\xi}[\kappa(\gamma_1(\eta_1(\xi),  w_1),\hdots,\gamma_N(\eta_N(\xi),  w_N),\xi)].
\]
Next we have
\[
\mathcal{J}_{\TO_C}^{*} = \min_{ \mathbb{P}}  \mathbb{E}_{w}[ \mathcal{J}_{\TO}^{*}(w)].
\]
It is linear program, thus it has optimal at extreme points, that is, $w^{*} = \arg \min_{w} \mathcal{J}_{\TO}^{*}(w).$ Then 
\[
\mathcal{J}_{\TO_C}^{*} =  \mathcal{J}_{\TO}^{*}(w*).
\]
Now consider that $\mathcal{J}_{\TO_C}$ is a convex function of decision rule $\gamma.$
If the decision rule is linear in its information, that is, $\gamma_i(\eta_i(\xi), w_i) = \alpha_{i1} \eta_i(\xi) + \alpha_{i2}w_i $, then clearly cost function will convex in $\alpha_{i1}$ and $\alpha_{i2}$ for all $ i = 1, \cdots , N.$ Without loss of generality assume that $\mathbb{E}[w_i] = 0 $ for all $i =1, \cdots , N.$ Since $w$ and $\xi$ are independent the cost function will be separable and  minimization w.r.t. variable $\alpha_{i1}$ and $\alpha_{i2}$ for all $ i = 1, \cdots , N.$ It implies that cost will be minimum iff $\alpha_{i2} = 0$ for all $i = 1, \cdots, N.$ Thus no weightage given to additional information under this decision rule.

\end{proof}
For LQG team problem in Appendix~\ref{sub:LQGteam-randomness}, it is  illustrated that  if private and common randomness independent of the environment $\xi$, it does not improve the expected cost function.

\subsection{Randomness dependent on $\xi$}
Consider a scenario where consultant provides an extra randomness about an environment to decision makers. That means these extra randomness is correlated with an environment $\xi$. 

Let $\omega = (\omega_1, \ldots, \omega_N)$ be a random vector represents an extra randomness provided to decision makers by consultant. Further assume that  $\omega $ is function of $\xi$, 
i.e. $\omega = f(\xi) = (f_1(\xi), \ldots, f_N(\xi))$, here $f, f_i$ be the measurable functions. The strategies of $\DM_i$ are
 $\gamma_i: y_i \times \omega_i \rightarrow u_i$, $\gamma_i \in \Gamma_i$ space of strategies and 
 $u_i \in U_i$ space of decision variables. The team optimal cost is defined as follows.
 \begin{equation}
 \min_{\gamma \in \Gamma} \mathcal{J}({\gamma}) = \min_{\gamma \in \Gamma} \mathbb{E}_{\xi, \omega}[\kappa(u_1 = \gamma_1(\eta_1(\xi), \omega_1 ), \ldots, \gamma_N(\eta_N(\xi), \omega_N))].
 \end{equation}
Note that  $\omega$ is function of $\xi.$ The optimal cost function is 
\begin{equation}
\mathcal{J}_{\TO_{ER}}^{*} = \min_{\gamma \in \Gamma} \mathbb{E}_{\xi}[\kappa(u_1 = \gamma_1(\eta_1(\xi), f_1(\xi) ), \ldots, \gamma_N(\eta_N(\xi), f_N(\xi)))].
\label{eq:teamoptcomdepxi}
\end{equation}
In distributed team problem with no extra randomness, decision maker have only partial observation about $\xi$. Thus an observations about $\xi$ is distributed among decision makers and  an optimal team cost $\mathcal{J}_{\TO}^{*}$ found in section \ref{subsec:team-prob}. When a consultant provides an extra randomness  about an environment $\xi$ to  the decision makers. 
Essentially, there is an increase in observation about $\xi$ available at decision makers. 
Intuitively, we expect that optimal cost under extra randomness in
distributed stochastic team will improve optimal  cost functional. Thus we have following result.
\begin{proposition}
In  distributed static stochastic team problem,
\begin{eqnarray}
\mathcal{J}_{\rm TO}^{*} \geq \mathcal{J}_{\TO_{ER}}^{*}.
\label{eq:teamproblemCR}
\end{eqnarray}
\label{lemma:staticteamproblem2}
\end{proposition}
\begin{proof}
We develop the proof using the ideas from \cite{Blackwell51}.  Let $\mathcal{B}_1 = \{\eta_1(\xi), \eta_2(\xi), \cdots, \eta_N(\xi) \} $ be the information available at team and  $\mathcal{B}_2 = \{(\eta_1(\xi), f_1(\xi)),(\eta_2(\xi),f_1(\xi)), \cdots, (\eta_N(\xi), f_N(\xi) ) \} $ be the another information available at team with extra common randomness. Thus $\mathcal{B}_1 \subset \mathcal{B}_2,$ i.e., $\mathcal{B}_2$ is more informative than $\mathcal{B}_1.$ Since DMs can cooperate and communicate in team problem. 

The  minimization problem $\min_{\gamma \in \Gamma}\mathbb{E}_{\xi}\left[\kappa \left(u_1, u_2, \cdots, u_N \right)~\vert ~\mathcal{B} \right].$ As $f_i$s are measurable functions, $\eta_i$s are measurable functions, so $\gamma_i$ are measurable and  $\Gamma$ is closed  bounded convex set. The cost function is also convex and measurable, thus from \cite[Theorem $2$]{Blackwell51} we can have 
\begin{equation*}
\min_{\gamma \in \Gamma}\mathbb{E}_{\xi}\left[\kappa \left(u_1, u_2, \cdots, u_N \right)~\vert ~\mathcal{B}_2 \right]  \leq \min_{\gamma \in \Gamma}\mathbb{E}_{\xi}\left[\kappa \left(u_1, u_2, \cdots, u_N \right)~\vert ~\mathcal{B}_1 \right]
\end{equation*}
This implies the desired result.
%
\end{proof}


Consider LQG stochastic team problem which has decision maker $\DM_1$ and $\DM_2$ in a team, and we have following different variation of LQG team problem based on types of observation available at decision makers.

{Problem $1$:}
 Let decision variable   $u_1  = Ay $, where $A$ is diagonal matrix, $\diag (A) = [\alpha_{11}, \ldots, \alpha_{N1}]$,  $y$ is observation available at decision makers, $y = [ y_1, y_2]^{T} = [\mu_1, \mu_2]^{T} $, and $\Sigma$ is covariance matrix of random vector $y$. The expected team cost
\begin{equation*}
\mathcal{J}_{\TO_{LQG,1}}^{*} = \min_{A} \mathbb{E}_{\xi}[y^{T}A^{T}BA y + 2 y^{T}A^{T}S \xi ] =   \min_{A} \Tr[A^{T}BA \Sigma + 2 A^{T}S \Sigma ].
\end{equation*}  
 
 Let $A^{*}$ be the matrix such that optimal cost function of team is
\begin{equation*}
\mathcal{J}_{\TO_{LQG,1}}^{*} =  \Tr[A^{*T}BA^{*} \Sigma + 2 A^{*T}S \Sigma ].
\end{equation*} 

{Problem $2$:}
Let decision variable   $u_2  = \tilde{A}\tilde{y} $, where $\tilde{A}$ is diagonal matrix, $\diag (\tilde{A}) = [\alpha_{11}, \ldots, \alpha_{N1}]$, $\tilde{y}$ is observation available at decision makers, $\tilde{y} = [ y_2, y_1]^{T} = [\mu_2, \mu_1]^{T} $.  
Note that 
\begin{equation*}
\tilde{y} = \left[ \begin{array}{c}
y_2   \\
y_1
\end{array} \right]  = \left[ \begin{array}{cc}
0 & 1   \\
1 & 0
\end{array} \right] 
\left[ \begin{array}{c}
y_1  \\
y_2
\end{array} \right] 
\end{equation*}
Let $\tilde{I} = \left[ \begin{array}{cc}
0 & 1   \\
1 & 0
\end{array} \right] $. Thus $\tilde{y} = \tilde{I} {y}$ and $\xi = y = \tilde{I} \tilde{y}$.

The expected team cost is
\begin{equation*}
\mathcal{J}_{\TO_{LQG,2}}^{*} = \min_{\tilde{A} } \mathbb{E}_{\xi}[\tilde{y}^{T}\tilde{A}^{T}B \tilde{A}  \tilde{y} + 2 \tilde{y}^{T}\tilde{A}^{T} S \xi ] =   \min_{\tilde{A} } \Tr[\tilde{A}^{T}B\tilde{A} \Sigma + 2 \tilde{A}^{T} \tilde{S} \tilde{\Sigma} ].
\end{equation*}
Here $\tilde{S} := S \tilde{I}$ and $\tilde{\Sigma}$ denote the covariance matrix of random vector $\tilde{y}$ . 

Let $A^{**}$ be the matrix such that 
\begin{equation*}
\mathcal{J}_{\TO_{LQG,2}}^{*} =  \Tr[A^{**T}BA^{**} \tilde{\Sigma} + 2 A^{**T}\tilde{S} \tilde{\Sigma} ].
\end{equation*}

{Problem $3$:}
Let decision variable $u_3 =  C \omega $, where $C$ is diagonal matrix, $\omega = [\omega_1, \omega_2]^{T}$, $\omega_1 = \beta y_1 + (1 - \beta) y_2$, $= \omega_2 =  (1- \beta) y_1 +  \beta y_2$,  $\beta \in (0,1)$. Hence $\omega = \beta y + (1 - \beta)\tilde{y} $.  We assume that decision maker has available common randomness provided by a consultant. These common randomness is convex combination of observation available at decision maker that is $y_1$ and $y_2$.
For example $\beta = \frac{1}{2}$, a  consultant provides an average of observations.   
The optimal cost functional  is 

\begin{equation*}
\mathcal{J}_{\TO_{LQG,3}}^{*} = \min_{u_3 \in U} \mathbb{E}_{\xi}[u_3^{T}Bu_3 + 2 u_3^{T}S\xi].
\end{equation*} 

\begin{proposition}
\begin{enumerate}
\item 
\begin{equation*}
\mathcal{J}_{\TO_{LQG,3}}^{*} \leq \beta \mathcal{J}_{\TO_{LQG,1}}^{*} + (1- \beta) \mathcal{J}_{\TO_{LQG,2}}^{*} .
\end{equation*}
 \item 
 If $\tilde{\Sigma} = \Sigma$ and $\tilde{S} = S$, then

\begin{equation*}
\mathcal{J}_{\TO_{LQG,1}}^{*} = \mathcal{J}_{\TO_{LQG,2}}^{*}.
\end{equation*}

Furthermore, $A^{*} = A^{**}$. 
Also, 
\begin{equation*}
\mathcal{J}_{\TO_{LQG,3}}^{*} \leq  \mathcal{J}_{\TO_{LQG,1}}^{*}. 
\end{equation*}
\end{enumerate}
\end{proposition}
\begin{proof}
1)
We have:
\begin{equation*}
\mathcal{J}_{\TO_{LQG,3}}^{*} = \min_{u_3 \in U} \mathbb{E}_{\xi}[u_3^{T}Bu_3 + 2 u_3^{T}S\xi]
\end{equation*}
Now, 
\begin{eqnarray}
\nonumber
 u_3^{T}Bu_3  & = & \omega^{T} C^{T} B C \omega  \\ 
 \nonumber
 & = & (\beta y + (1 - \beta)\tilde{y})^{T} C^{T} B C (\beta y + (1 - \beta)\tilde{y})\\
 & \leq & \beta y^{T} C^{T} B C y + (1 - \beta) \tilde{y}^{T} C^{T} B C \tilde{y}
 \label{eq:LQGcomrnd}
\end{eqnarray}
Since $B$ is symmetric positive definite matrix, $ (\beta y + (1 - \beta)\tilde{y})^{T} C^{T} B C (\beta y + (1 - \beta)\tilde{y})$  is quadratic convex function. Thus
inequality in \eqref{eq:LQGcomrnd} follows from convexity property of function.
\begin{eqnarray}
\nonumber
\mathcal{J}_{\TO_{LQG,3}}^{*} & \leq & \min_{C} \mathbb{E}_{\xi}[ \beta y^{T} C^{T} B C y + (1 - \beta) \tilde{y}^{T} C^{T} B C \tilde{y} +  2 \beta y^{T} C^{T}S \xi + 2 (1 - \beta) \tilde{y}^{T} C^{T} S \xi] \\ 
\nonumber
& = & \min_{C} \Tr[\beta C^{T} B C \Sigma + 2 \beta C^{T}S \Sigma + (1 - \beta)C^{T} B C \tilde{\Sigma}  +  2(1 - \beta) C^{T} \tilde{S} \tilde{\Sigma} ] \\
\nonumber
& = &  \beta  \min_{C} \Tr[C^{T} B C \Sigma + 2  C^{T}S \Sigma ]  + (1 - \beta)  \min_{C} \Tr[C^{T} B C \tilde{\Sigma}  +  2C^{T} \tilde{S} \tilde{\Sigma} ]  \\ \nonumber 
& = &   \beta \mathcal{J}_{\TO_{LQG,1}}^{*} + (1- \beta) \mathcal{J}_{\TO_{LQG,2}}^{*} .
\end{eqnarray}

2) 

Let $\tilde{\Sigma} = \Sigma$ and $\tilde{S} = S$, we have:

\begin{equation*}
\mathcal{J}_{\TO_{LQG,1}}^{*} =  \min_{A} \Tr[A^{T}BA \Sigma + 2 A^{T}S \Sigma ].
\end{equation*}

\begin{equation*}
\mathcal{J}_{\TO_{LQG,2}}^{*} = \min_{\tilde{A} } \Tr[\tilde{A}^{T}B\tilde{A} \Sigma + 2 \tilde{A}^{T} \tilde{S} \tilde{\Sigma} ].
\end{equation*}
Clearly, $\mathcal{J}_{\TO_{LQG,1}}^{*} = \mathcal{J}_{\TO_{LQG,2}}^{*}$. 
Consequntly, $A^{*} = A^{**}$. Hence,
\begin{equation*}
\mathcal{J}_{\TO_{LQG,3}}^{*} \leq \mathcal{J}_{\TO_{LQG,1}}^{*}.
\end{equation*}

\end{proof}

So far, we studied role of common randomness (information) in a team problem. In next section, we describe the role of common randomness in two team zero-sum game.

\section{Private and common randomness in static team vs team zero-sum game}
\label{sec:teamvsteamprivatecom}
We study role of private and common randomness in static two-team zero-sum game.
We compare the static LQG team with zero-sum LQG team game under private and common randomness. Then We demonstrate the two team zero-sum discrete game.

Now consider the case where there are $ N+M $ DMs. Let $ \Mscr = \{N+1,\hdots,M\} $. $ \DM_i, i \in \Nscr $ 
comprise of a single team, say Team 1, and $ \DM_j, j \in \Mscr $ comprise of Team 2. Team 1 and Team 2 play a zero-sum game. 
Let $ u =(u_1,\hdots,u_N), \gamma=(\gamma_1,\hdots,\gamma_N)$ denote the actions of players of Team 1 and $ v=(v_{N+1},\hdots,v_M), \delta=(\delta_{N+1},\hdots,\delta_M) $ denote the actions of players in Team 2. 
Suppose the function the teams want to optimize is \[ \min_{u_i = \gamma_i(y_i),  i \in \Nscr} \max_{v_j = \delta_j(y_j), j \in \Mscr}\Ebb[\kappa(u,v,\xi)] \]

\begin{theorem}
If the zero-sum team game admits a saddle point, randomness independent of $ \xi $ does not benefit either team.
\end{theorem}
\begin{proof}
We have:
\[ \min_{u_i = \gamma_i(y_i),  i \in \Nscr} \max_{v_j = \delta_j(y_j), j \in \Mscr}\Ebb[\kappa(u,v,\xi)] =  \max_{v_j = \delta_j(y_j), j \in \Mscr}\min_{u_i = \gamma_i(y_i),  i \in \Nscr} \Ebb[\kappa(u,v,\xi)]
 \]
\begin{align}
\min_{u_i = \gamma_i(y_i),  i \in \Nscr} \max_{v_j = \delta_j(y_j), j \in \Mscr}\Ebb[\kappa(u,v,\xi)] &\geq \min_{u_i = \gamma_i(y_i,w),  i \in \Nscr} \max_{v_j = \delta_j(y_j,z), j \in \Mscr}\Ebb[\kappa(u,v,\xi)] \nonumber \\
& \geq  \max_{v_j = \delta_j(y_j,z), j \in \Mscr} \min_{u_i = \gamma_i(y_i,w),  i \in \Nscr}\Ebb[\kappa(u,v,\xi)] \label{eq:minmax} \\
& \geq \max_{v_j = \delta_j(y_j), j \in \Mscr}\min_{u_i = \gamma_i(y_i),  i \in \Nscr} \Ebb[\kappa(u,v,\xi)]. \nonumber 
\end{align} 
Eq \eqref{eq:minmax} follows from:
\[  \max_{v_j = \delta_j(y_j,z), j \in \Mscr}\Ebb[\kappa(u,v,\xi)]  \geq  \min_{u_i = \gamma_i(y_i,w),  i \in \Nscr} \Ebb[\kappa(u,v,\xi)]  \].
Consequently,
\[  \min_{u_i = \gamma_i(y_i,w),  i \in \Nscr} \max_{v_j = \delta_j(y_j,z), j \in \Mscr}\Ebb[\kappa(u,v,\xi)]  \geq  \max_{v_j = \delta_j(y_j,z), j \in \Mscr}\min_{u_i = \gamma_i(y_i,w),  i \in \Nscr} \Ebb[\kappa(u,v,\xi)]  \].
\end{proof}
\begin{theorem}
If zero-sum game admits a saddle point, common randomness dependent of $\xi$ is provided to one of team, then that team benefits. Suppose the consultant provides common randomness $z$ which is dependent of $\xi$ to decision makers of a team  say, Team $2.$ Then we want to optimize
\begin{equation*}  
 \mathcal{J}_{\mathrm{TO}_{ZS,CR}} =   \min_{u_i = \gamma_i(y_i),  i \in \Nscr} \max_{v_j = \delta_j(y_j, z), j \in \Mscr}\Ebb[\kappa(u,v,\xi)].
\end{equation*}
Further, $ \mathcal{J}_{\mathrm{TO}_{ZS,CR}}  =  \mathcal{J}_{\mathrm{TO}_{ZS}},$ where
\[
\mathcal{J}_{\mathrm{TO}_{ZS}} = \min_{u_i = \gamma_i(y_i),  i \in \Nscr} \max_{v_j = \delta_j(y_j), j \in \Mscr}\Ebb[\kappa(u,v,\xi)].
\]
\end{theorem}
\begin{proof}

We know from a team decision problem with common randomness dependent of $\xi,$ then 
\begin{align*}
\max_{v_j = \delta_j(y_j, z), j \in \Mscr}\Ebb[\kappa(u,v,\xi)]
\geq  \max_{v_j = \delta_j(y_j), j \in \Mscr}\Ebb[\kappa(u,v,\xi)] 
\end{align*}
\begin{align*}
\min_{u_i = \gamma_i(y_i),  i \in \Nscr} \max_{v_j = \delta_j(y_j, z), j \in \Mscr}\Ebb[\kappa(u,v,\xi)]
 \geq  \min_{u_i = \gamma_i(y_i),  i \in \Nscr} \max_{v_j = \delta_j(y_j), j \in \Mscr}\Ebb[\kappa(u,v,\xi)]
\end{align*}
Since we assume saddle point solution of zero-sum game, 
\[
\min_{u_i = \gamma_i(y_i),  i \in \Nscr} \max_{v_j = \delta_j(y_j, z), j \in \Mscr}\Ebb[\kappa(u,v,\xi)] = \max_{v_j = \delta_j(y_j,z), j \in \Mscr}\min_{u_i = \gamma_i(y_i),  i \in \Nscr} \Ebb[\kappa(u,v,\xi)]
\]
We also have
\[
\max_{v_j = \delta_j(y_j,z), j \in \Mscr}\min_{u_i = \gamma_i(y_i),  i \in \Nscr} \Ebb[\kappa(u,v,\xi)] 
\geq \max_{v_j = \delta_j(y_j), j \in \Mscr}\min_{u_i = \gamma_i(y_i),  i \in \Nscr} \Ebb[\kappa(u,v,\xi)] 
\]
If two-team zero sum game without common randomness admits a saddle point, then
\[
\max_{v_j = \delta_j(y_j), j \in \Mscr}\min_{u_i = \gamma_i(y_i),  i \in \Nscr} \Ebb[\kappa(u,v,\xi)]  =
\min_{u_i = \gamma_i(y_i),  i \in \Nscr} \max_{v_j = \delta_j(y_j), j \in \Mscr}\Ebb[\kappa(u,v,\xi)].
\]
Hence result $ \mathcal{J}_{\mathrm{TO}_{ZS,CR}}  =  \mathcal{J}_{\mathrm{TO}_{ZS}}$  follows.
\end{proof}

\begin{remark}
\begin{itemize} 
\item 	
If the common or private information is uncorrelated with an environment or uncertainty of world, no one can gain anything from this information in team vs team zero-sum game. This is also illustrated numerically for LQG zero sum team vs team game is illustrated in Appendix~\ref{app:common-info-indep-LQG}. 
\item In next subsection, we describe that a team having private information correlated with environment benefits. This implies that the team with more information manage to decrease the cost and even this is true in LQG teams decision problem. This is first observed by \cite{Witsenhausen71} and later this is extended to LQG teams problem in \cite{Ho74}.
\item   We present  results in our stochastic team vs team zero sum game. We illustrate role of common randomness in team vs team LQG zero sum game by numerical examples in Appendix~\ref{app:common-randomness-depend-envi-LQG}.
\end{itemize}
\end{remark}

\subsection{Role of private randomness dependent on $\xi$}
 
Let $y_i = \eta_i(\xi)$ be the information available at player $i,$ and $\widetilde{y}^1 = (y_1, y_2, \cdots, y_N)$ be information available at team $1$ and $\widetilde{y}^2 = (y_{N+1}, y_{N+2}, \cdots, y_{N+M})$ be the information available at team $2.$ Note that a team $1$ is minimizing using control $u$ and team $2$ is maximizing with control $v.$

Define the cost function  
\begin{equation*}
J(u,v) = \Ebb[\kappa(u,v,\widetilde{y}^1, \widetilde{y}^2,\xi)] 
\end{equation*}
From saddle point condition at the information structure $(\widetilde{y}^1,\widetilde{y}^2),$ we have 
\begin{eqnarray*}
J(u^*,v)  \leq J(u^*,v^*) \leq J(u,v^*). 
\end{eqnarray*}
The optimal decision pair is $(\overline{u}^*,v^*)$ at the information structure $\widetilde{y}^1,$ and $\widetilde{y}^2.$
Similarly, one can define saddle point condition for null  information structure and has only prior knowledge about $\xi,$  information structure 
is $(\overline{y}^1,\overline{y}^2) $  and optimal decision pair is  $({u}^0,{v}^0).$

 The value of information for team $1$ and team $2$ is defined as follows. 
\begin{eqnarray*}
V_1\left(\widetilde{y}^1, \widetilde{y}^2 \right) &=& J(u^*,v^*) - J({u}^0,{v}^0) \\
V_2\left(\widetilde{y}^1, \widetilde{y}^2 \right) &=& -V_1\left(\widetilde{y}^1, \widetilde{y}^2 \right)
\end{eqnarray*}
Suppose the information at a team, say team $2$ is fixed, i.e., $\eta^{\prime}_i(\xi) = \eta_i(\xi)$ for $i = N+1, \cdots, N+M.$ 
The opponent gets more information, say team $1,$ i.e., $\eta^{\prime}_i(\xi) \subseteq \eta_i(\xi)$ for $i=1,2, \cdots, N.$ Thus the decision set for team  $1$ is $\mathcal{A}_{\eta^{\prime}} \subseteq \mathcal{A}_{\eta}$ and that for team $2$ is  $\mathcal{C}_{\eta^{\prime}} = \mathcal{C}_{\eta}.$ We have the following result. 
\begin{lemma}
	If the information of team $1$ is increasing, i.e., $\eta^{\prime}_i(\xi) \subseteq \eta_i(\xi)$ for $i=1,2, \cdots, N,$ and the information of team $2$ is fixed, i.e., $\eta^{\prime}_i(\xi) = \eta_i(\xi)$ for $i = N+1, \cdots, N+M,$ then the value of information satisfy the following inequality
	\begin{equation*}
	V_1\left(\widetilde{y}^1, \widetilde{y}^2 \right) \leq V_1\left(\widehat{y}^1, \widehat{y}^2 \right).
	\end{equation*}
Here  $y_i = \eta_{i}(\xi),$ $\widetilde{y}^1  = (y_1, \cdots, y_N),$    
$\widetilde{y}^2  = (y_{N+1}, \cdots, y_{N+M}),$ and 
$y_i^{\prime} = \eta_{i}^{\prime}(\xi),$ $\widehat{y}^1  = (y_1^{\prime}, \cdots, y_N^{\prime}),$    
$\widehat{y}^2  = (y_{N+1}^{\prime}, \cdots, y_{N+M}^{\prime}).$ 

\label{lemma:LQG-Additional-Info-teamvsteam}
\end{lemma}
The proof is analogous to \cite[Lemma $3.3$]{Ho74}. For clarity purpose we provide details is as follows. 

The saddle point condition at information structure $\eta(\xi)$ implies that 
\begin{eqnarray}
	J(u^*,v)  \leq J(u^*,v^*) \leq J(u,v^*) 
	\label{eqn:saddle-point-cond-1}
\end{eqnarray}
for $u \in \mathcal{A}_{\eta}, v \in \mathcal{C}_{\eta}.$ 
Another saddle point condition at information structure $\eta^{\prime}(\xi)$ is 
\begin{eqnarray}
	J(\widehat{u}^{*},\widehat{v})  \leq J(\widehat{u}^{*},\widehat{v}^{*}) \leq J(\widehat{u},\widehat{v}^{*}) 
	\label{eqn:saddle-point-cond-2}
\end{eqnarray}
for $\widehat{u} \in \mathcal{A}_{\eta^{\prime}}$ and $\widehat{v}\in \mathcal{C}_{\eta^{\prime}}.$

Since $\mathcal{C}_{\eta} = \mathcal{C}_{\eta^{\prime}}$  we can have 
$\widehat{v}^{*} \in \mathcal{C}_{\eta}$ and then it implies that 
\begin{eqnarray}
J(u^*,v^*)  = J(u^*,\widehat{v}^{*}).
\end{eqnarray}
Because  $\mathcal{A}_{\eta^{\prime}} \subseteq \mathcal{A}_{\eta},$ and $\widehat{u}^* \in \mathcal{A}_{\eta^{\prime}}$ implies $ \widehat{u}^* \in \mathcal{A}_{\eta}.$ Further, 
\begin{eqnarray}
 J(\widehat{u},\widehat{v}^{*}) \geq J(\widehat{u}^*,\widehat{v}^{*}) \geq 
 J({u}^*,\widehat{v}^{*}) = J({u}^*,{v}^{*})
\end{eqnarray}

Thus we get $J(\widehat{u}^*,\widehat{v}^{*}) \geq J({u}^*,{v}^{*}).$ As we note that $J(u^0, v^0)$ does not change. After substracting $J(u^0, v^0),$ we have desired inequality 
\begin{equation*}
V_1\left(\widetilde{y}^1, \widetilde{y}^2 \right) \leq V_1\left(\widehat{y}^1, \widehat{y}^2 \right).
\end{equation*}

\subsection{Discrete team vs team zero-sum  game}
In this section, we investigate discrete  team vs team zero-sum game 
and the role of extra randomness in the team and its decision makers. 
\begin{claim}
\label{thm:two-teamzero-sumgame}
In  discrete team vs team zero-sum game,
\begin{enumerate}
\item 

it may not admit pure-strategy saddle point solution, 
\item if a coordinator provides the private randomness independent of an environment to decision makers  of team then it  benefit both team and improves the team cost. But it may not achieve Nash equilibrium,
\item if a consultant provides the common randomness to decision makers of team, then it lead to  improve in team cost. But it may not have Nash equilibrium.
\end{enumerate}
\end{claim}
Proofs of these are difficult to obtain but we provide examples in appendix \ref{proof:two-teamzero-sumgame} to support our claim.

%
%
%

\section{Discussion and Conclusions}
\label{sec:discussion-concluson}

The value of information is classic problem in decision theory. As information increases, we anticipated that the optimal cost decreases. This is first illustrated for statistical problems in \cite{Blackwell51}. 
In stochastic team problem and stochastic team vs team zero sum games, the value of private information to decision makers is not  explicitly presented in earlier literature.

We analyzed a stochastic  team decision problem when decision makers are provided with external private randomness which is correlated or independent of environment. 
The private randomness independent  of environment does not decrease the cost function.  But this randomness dependent on environment  provided to DMs in a team  decreases the cost function of team compare to no randomness. In stochastic LQG team decision problem under special information structure, we have shown that the correlated randomness decreases the cost function. 

We next studied stochastic team vs team zero sum game, and showed that the randomness independent of environment does not benefit either time if a game admits a saddle point condition. In LQG team vs team zero sum game, we analyze the role of common randomness which is correlated with environment for one of team, then the optimal value function decreases with information. We further extended this finding to discrete  team vs team zeros sum game when there is no saddle point condition and observed that common or private randomness independent of environment  benefits both team. Even common randomness dependent on environment benefit a team and improves cost. This may not lead to saddle point condition.  

It opens future research direction on problem of role of private or common randomness  in stochastic teams with non zero sum games and sequential stochastic dynamic teams. 
Another research directions is on correlated equilibrium behaviors and common knowledge in  sequential stochastic team vs team games.

%
%

\section{ACKNOWLEDGMENT}
Most part of this work was carried out at the Bharti Centre for Communications at IIT Bombay  and EE Dept. IIT Bombay, where author was PhD scholar.   Part of this work  was done  at EE Dept. IIT Madras during Postdoctoral Fellowship. Author is very grateful to Prof. Ankur Kulkarni, SC Dept. IIT Bombay for  guidance and extensive discussion on problem of Team decision theory and pointing out  references.  Author is thankful to Prof. D. Manjunath, EE Dept. IIT Bombay for initial support on the work. Author is also thankful to ECE Dept. IIIT Allahabad for financial support.

\bibliographystyle{IEEE}
\bibliography{TDT}

\newpage 
\appendix

\input{LQG-team}

\input{discrete-zero-sum-teamgame}

\input{LQG-team-vs-team-zerosum}

\end{document}

%% file: macros-ieee.tex
\usepackage{amsmath, amssymb, bbm, xspace}
\usepackage{epsfig}
\usepackage{longtable}
\usepackage{color}
\usepackage{mathrsfs}
\usepackage{subfigure}

\usepackage{courier}

\newtheorem{theorem}{Theorem}[section]

\newtheorem{lemma}{Lemma}[section]
\newtheorem{claim}[theorem]{Claim}
\newtheorem{proposition}[theorem]{Proposition}

\def\bkE{{\rm I\kern-.17em E}}
\def\bk1{{\rm 1\kern-.17em l}}
\def\bkD{{\rm I\kern-.17em D}}
\def\bkR{{\rm I\kern-.17em R}}
\def\bkP{{\rm I\kern-.17em P}}

\def\bkZ{{\bf{Z}}}

\def\bkE{{\rm I\kern-.17em E}}
\def\bk1{{\rm 1\kern-.17em l}}
\def\bkD{{\rm I\kern-.17em D}}
\def\bkR{{\rm I\kern-.17em R}}
\def\bkP{{\rm I\kern-.17em P}}

\def\bkZ{{\bf{Z}}}
\def\b12{(\beta_1,\beta_2)}

\newcounter{example}
\renewcommand{\theexample}{\thesection.\arabic{example}}

\newcounter{remark}
\renewcommand{\theremark}{\thesection.\arabic{remark}}

\newenvironment{remark}{{\noindent \it Remark: }}{\hfill $\square$}

\def\Ebb{\mathbb{E}}
\newlength{\noteWidth}
\long\def\notes#1{\ifinner
{\tiny #1}
\else
\marginpar{\parbox[t]{\noteWidth}{\raggedright\tiny #1}}
\fi\typeout{#1}}

 \def\notes#1{\typeout{read notes: #1}} 

\def\mi{^{-i}}

\newcommand{\Real}{\ensuremath{\mathbbm{R}}}

\def\Nbb{{\mathbb{N}}}

\def\ast{\buildrel{t}\over\rightarrow}

\def\diag{\mathop{\hbox{\rm diag}}}

\def\spose#1{\hbox to 0pt{#1\hss}}

\def\text #1{\hbox{\quad#1\quad}}

\def\nthinsp{\mskip -2   mu}

\def\superstar{^{\raise 0.5pt\hbox{$\nthinsp *$}}}
\def\SUPERSTAR{^{\raise 0.5pt\hbox{$*$}}}

\def\lamstarT {\lambda^{\raise 0.5pt\hbox{$\nthinsp *$}T}}

\def\Jscr{{\cal J}}

\def\Mscr{{\cal M}}

\def\Pscr{{\cal P}}
\def\Qscr{{\cal Q}}

\def\Mscr{{\cal M}}
\def\Nscr{{\cal N}}

\let\forallnew\forall
\renewcommand{\forall}{\forallnew\ }
\let\forall\forallnew

		\def\bkE{{\rm I\kern-.17em E}}
		\def\bk1{{\rm 1\kern-.17em l}}
		\def\bkD{{\rm I\kern-.17em D}}
		\def\bkR{{\rm I\kern-.17em R}}
		\def\bkP{{\rm I\kern-.17em P}}
		\def\bkY{{\bf \kern-.17em Y}}
		\def\bkZ{{\bf \kern-.17em Z}}
		\def\bkC{{\bf  \kern-.17em C}}

{\begin{list}{}%
         {\setlength{\leftmargin}{#1}}%
         \item[]%
}
{\end{list}}

		\def\bsp{\begin{split}}
		\def\beq{\begin{eqnarray}}
		\def\bal{\begin{align*}}
		\def\bc{\begin{center}}
		\def\be{\begin{enumerate}}
		\def\bi{\begin{itemize}}
		\def\bs{\begin{small}}
		\def\bS{\begin{slide}}
		\def\ec{\end{center}}
		\def\ee{\end{enumerate}}
		\def\ei{\end{itemize}}
		\def\es{\end{small}}
		\def\eS{\end{slide}}
		\def\eeq{\end{eqnarray}}
		\def\eal{\end{align*}}
		\def\esp{\end{split}}

	\def\cp2problem#1#2#3#4{\fbox
		 {\begin{tabular*}{0.9\textwidth}
			{@{}l@{\extracolsep{\fill}}l@{\extracolsep{6pt}}l@{\extracolsep{\fill}}c@{}}
				#1 & & $#4 $ 
			\end{tabular*}}}

		\def\bkE{{\rm I\kern-.17em E}}
		\def\bk1{{\rm 1\kern-.17em l}}
		\def\bkD{{\rm I\kern-.17em D}}
		\def\bkR{{\rm I\kern-.17em R}}
		\def\bkP{{\rm I\kern-.17em P}}
		
		\def\bkZ{{\bf{Z}}}

\newcommand {\beeq}[1]{\begin{equation}\label{#1}}
\newcommand {\eeeq}{\end{equation}}
\newcommand {\bea}{\begin{eqnarray}}
\newcommand {\eea}{\end{eqnarray}}

\def\texitem#1{\par\smallskip\noindent\hangindent 25pt
               \hbox to 25pt {\hss #1 ~}\ignorespaces}



%% file: LQG-team.tex
\subsection{LQG Team Problem}
\label{sub:LQGteam-randomness}
Now we examine an example of a LQG team problem. 

Consider a LQG team problem of having $N$ decision maker.
Let an environment $\xi := [\mu_1, \ldots,\mu_N ]^{T}$ be random vector; it is  Gaussian distributed zero mean and covariance $\Sigma.$  Let $y_i = \eta_{i}(\xi)$ be the information observed by $\DM_i$, 
$y = [ y_1, \ldots, y_N]^{T}$ information vector observed by decision makers. 
In a static LQG team problem optimal action is linear in information observed by decision maker. Thus action of $DM_i$ is $u_i = \gamma_{i}(y_i) = \alpha_{i1} y_i$. Then
\begin{eqnarray*}
u = (u_1, \ldots, u_N)^{T} = A y,
\end{eqnarray*}
where $A $ is  diagonal matrix  of  dimensional $N \times N$, $\diag(A) = [\alpha_{11}, \ldots, \alpha_{N1}]$.
Standard LQG problem assumes cost function to be quadratic in nature. 
The cost function is $\kappa(u, \xi) := u^{T}Bu + 2 u^{T}S\xi,$
here $B$ is symmetric positive matrix. 

The team optimal solution of LQG team problem is $\gamma \in \Gamma$ such that
\begin{equation}
\mathcal{J}_{{\rm TO_{LQG}}}^{*} \triangleq \min_{\gamma \in \Gamma}\mathcal{J}(\gamma) = \min_{u
 \in U} \mathbb{E}_{\xi}[\kappa(u,\xi)]  = \min_{u \in U} \mathbb{E}_{\xi} 
[u^{T}Bu + 2 u^{T}S\xi] . 
\label{eq:teamoptLQG}
\end{equation}
Replacing $u = Ay$, we obtain
\begin{equation*}
\mathcal{J}_{{\rm TO_{LQG}}}^{*}  = \min_{A} \mathbb{E}_{\xi} 
[y^{T}A^{T}BA y + 2 y^{T}A^{T}S\xi].
\end{equation*}
Further this can expressed as deterministic optimization problem as follows.
\begin{equation*}
\mathcal{J}_{{\rm TO_{LQG}}}^{*} = \min_{A}\Tr[A^{T}BA \Sigma + 2 A^{T}S \Sigma ],
\label{eq:teamoptdetLQG}
\end{equation*}
Note that $\Tr$ denote trace of matrix.
\subsubsection{Private randomness independent of $\xi$ }
We will show that in LQG team problem the private randomness provided by a coordinator  do not benefit the  team optimal cost functional.

Consider $\omega = [\omega_1, \ldots, \omega_N]^{T}$
is private randomness available to decision makers, it is Gaussian distributed with zero mean and covariance matrix $\Sigma_{1}$ is diagonal; $\omega_i$ is private randomness available at $\DM_i$. We suppose that $\omega_i$ is independent of $\omega_j$ for $i \neq j$ and it is also independent of $y$. ($\mathbb{E}[\omega_i \omega_j] = 0$ for $i \neq j$ and $\mathbb{E}[\omega_i y_k] = 0$ for $i \neq  k$, $1 \leq i,j, k \leq N$.) 

The action $u_i = \gamma_i(y_i, \omega_i) = \alpha_{i1} y_i + \alpha_{i2} \omega_i$. Let $u = Ay + C \omega$, where
$A $ and $C$ are diagonal matrix of dimension $N \times N$, $\diag(A) = [\alpha_{11}, \ldots, \alpha_{N1}]$ and $\diag(C) = [\alpha_{12}, \ldots, \alpha_{N2} ]$.

The optimal expected cost functional of LQG team problem with private randomness is  
\begin{equation}
\Jscr^*_{\TO_{P,LQG}} \triangleq \min_{Q \in \Qscr}\mathbb{E}_{Q}[\kappa(u,\xi)]
= \min_{A} \Tr[A^{T}BA \Sigma + 2 A^{T}S \Sigma] + \min_{C} \Tr[C^{T}BC\Sigma_1]
. \label{eq:teamoptqLQG-A}
\end{equation}
From equation \eqref{eq:teamoptqLQG-A}, $\min_{C} \Tr[C^{T}BC\Sigma_1] =  0$ if and only if  $C$ is zero matrix. Hence $\Jscr^*_{\TO_{LQG}} = \Jscr^*_{\TO_{P,LQG}}.$

\subsubsection{Common randomness independent of $\xi$ }
We study a LQG team problem with common randomness has structure similar to that of LQG team problem with private randomness. We demonstrate that common randomness provided to decision makers by the consultant is independent of $\xi$, then it do not improve the expected cost functional.

Consider $\omega = [\omega_1, \ldots, \omega_N]^{T}$
is common randomness available to decision makers, it is Gaussian distributed with zero mean and covariance matrix $\Sigma_{2}$; $\omega_i$ is the common randomness at $\DM_i$. We suppose that $\omega_i$ is perfect correlation  with $\omega_j$ for $i \neq j$ and it is also independent of $y$. ($\mathbb{E}[\omega_i \omega_j] \neq 0$ for $i \neq j$ and $\mathbb{E}[\omega_i y_k] = 0$ for $i \neq  k$, $1 \leq i,j, k \leq N$.) The action $u_i = \gamma_i(y_i, \omega_i) = \alpha_{i1} y_i + \alpha_{i2} \omega_i$. Let $u = Ay + C \omega$. 
The optimal expected cost function is 
\begin{equation}
\mathcal{J}_{\TO_{C, LQG}}^{*} = \min_{A} \Tr[A^{T}BA \Sigma + 2 A^{T}S \Sigma] + \min_{C} \Tr[C^{T}BC\Sigma_2]
\label{eq:teamoptqcomLQG}
\end{equation}
Note that in LQG team problem, $B$ is symmetric positive definite matrix.
From \eqref{eq:teamoptqcomLQG}, expression $\min_{C} \Tr[C^{T}BC\Sigma_2]$ attains minimum value $ = 0$ if $C$ is zero matrix. Thus we have following relation,
$\Jscr^*_{\TO_{LQG}} = \Jscr^*_{\TO_{C,LQG}}$.
\subsubsection{Common randomness dependent on $\xi$}
Next, we demonstrate the result in \eqref{eq:teamproblemCR} via an example of LQG team problem.
Further we show numerically for two decision maker LQG team problem that there is strict inequality between team optimal cost with and without extra randomness, that is $\mathcal{J}_{\rm TO_{LQG}}^{*} > \mathcal{J}_{\TO_{ER,LQG}}^{*}$. 

Consider a LQG team problem consists of an environment  $\xi = [\mu_1, \ldots, \mu_N]^{T}$ as random vector with mean zero and covariance matrix $\Sigma$. The information observed by $\DM_i$ is $y_i = \eta_i(\xi) = \mu_i$, $y = [y_1, \ldots, y_N]^{T}$. Let $\omega = [\omega_1, \ldots, \omega_N]^{T}$ be the extra randomness provided by a consultant to decision makers. Furthermore, assume that $\omega = f(\xi)$ and $f$ is linear function in $\xi$. Thus $\omega_i = \sum_{j}^{N} \phi_{ij} \mu_j$, $\omega = \Phi \xi = \Phi y$, $\Phi$ is matrix of dimension $N \times N$, with entries in $\phi_{ij} \geq 0$ and $\sum_{j = 1}^{N} \phi_{ij} = 1$. The cost function is $\kappa(u, \xi) := u^{T}Bu + 2 u^{T}S\xi,$ the optimal expected cost under extra randomness is 
\begin{equation*}
\mathcal{J}_{\TO_{ER}}^{*} = \min_{u \in U} \mathbb{E}_{\xi}[u^{T}Bu + 2 u^{T}S\xi].
\end{equation*}
Since it is static LQG team problem, optimal decision rule is linear in observation variable. We assume that $ u_i = \alpha_{i1} y_i + \alpha_{i2} \omega_i,$ $u = Ay +C \omega$, where $A$ and $C$ are diagonal matrices, $\diag(A) = [\alpha_{11}, \ldots, \alpha_{N1}]$, $\diag(C) = [\alpha_{12}, \ldots, \alpha_{N2}]$.
The optimal expected cost is 
\begin{equation*}
\mathcal{J}_{\TO_{ER,LQG}}^{*}  = \min_{A, C} \mathbb{E}_{\xi} [y^{T}A^{T}BAy + 2 y^{T}A^{T} S \xi +  2 y^{T} \Phi^{T} C^{T} B A y + y^{T} \Phi^{T} C^{T} B C \Phi y + 2 y^{T} \Phi^{T} C^{T} S \xi].
\end{equation*}
We have $\xi \sim N(0, \Sigma)$, taking expectation and rewriting above expression, we obtain deterministic optimization problem as follows.
\begin{equation}
\mathcal{J}_{\TO_{ER,LQG}}^{*}  = \min_{A, C} \Tr[A^{T}BA^{T} \Sigma + 2 A^{T} S \Sigma +  2  \Phi^{T} C^{T} B A \Sigma + \Phi^{T} C^{T} B C \Phi \Sigma + 2  \Phi^{T} C^{T} S \Sigma].
\label{eq:staticteamlqgCR}
\end{equation}

 Intuitively, in LQG team problem with no extra randomness can described as  \textit{incomplete information} static LQG team problem.  Since extra randomness is linear function of an environment and under assumption of  nonzero linear coefficient  ($\phi_{ij} \neq 0$ for all $ 1 \leq i,j \leq N$),  LQG team problem with extra randomness can be describe as \textit{complete information} static LQG team problem. Thus  it is natural to expect that $\mathcal{J}_{\TO_{ER}}^{*} < \mathcal{J}_{\TO}^{*}$. But showing this result analytically difficult due to in-separability of  optimization problem \eqref{eq:staticteamlqgCR} into optimization problem with respect to $A$ and $C$.

To support our claim of $\mathcal{J}_{\TO_{ER}}^{*} < \mathcal{J}_{\TO}^{*}$, we  numerically evaluate the optimal cost functional with and without extra randomness which is dependent on $\xi$ for LQG two team problem and show that our claim is indeed true. Further, we show impact of correlation coefficient $\{\phi_{ij},  1 \leq i,j \leq 2 \}$  on optimal cost functional.
%
%

\subsubsection{Numerical example--LQG team problem}
Let $\xi = [\mu_1, \mu_2]^{T}$ denote the state of nature or an environment having probability distribution $N(0, \Sigma)$. 
Let $y_i = \eta_{i}(\xi) = \mu_i$ be the information observed at $\DM_i$ for $1 \leq i \leq 2.$
Let $\omega = [\omega_1, \omega_2]^{T}$ be an extra randomness provided by a consultant to decision makers. Consider $\omega_i = \phi_{i1} y_1 + \phi_{i2} y_2 $,  $u_i = \alpha_{i1} y_i + \alpha_{i2} \omega_i$ for $1 \leq i \leq 2$. Thus we have $A = \left[ \begin{array}{cc}
\alpha_{11} & 0  \\
0 & \alpha_{21} 
\end{array} \right] $,
$C =  \left[ \begin{array}{cc}
 \alpha_{12} & 0  \\
 0 & \alpha_{22}
\end{array} \right]$. Team optimal cost from \eqref{eq:staticteamlqgCR} is 
\begin{equation*}
\mathcal{J}_{\TO_{ER}}^{*}  = \min_{A, C} \Tr[A^{T}BA \Sigma + 2 A^{T} S \Sigma +  2  \Phi^{T} C^{T} B A \Sigma + \Phi^{T} C^{T} B C \Phi \Sigma + 2  \Phi^{T} C^{T} S \Sigma].
\end{equation*}
In this example, we suppose $B = \left[ \begin{array}{cc}
2 & -1   \\
-1 & 1
\end{array} \right]$,
$S = \left[ \begin{array}{cc}
1 & 0   \\
0 & 1
\end{array} \right]$,
$\Sigma =  \left[ \begin{array}{cc}
 \sigma_{\mu_1}^{2} & \sigma_{\mu_1, \mu_2}^{2}   \\
\sigma_{\mu_1, \mu_2}^{2} & \sigma_{\mu_2}^{2} 
\end{array} \right]$.
%



We define 
$ \delta_1 = \mathbb{E}[y_1 w_1] = \phi_{11} \sigma_{\mu_1}^{2} + \phi_{12} \sigma_{\mu_1 \mu_2}^{2}$, 
$\delta_2 = \mathbb{E}[y_1 w_2] = \phi_{21} \sigma_{\mu_1}^{2} + \phi_{22} \sigma_{\mu_1 \mu_2}^{2} $, 
$\delta_3 = \mathbb{E}[y_2 w_1] = \phi_{11} \sigma_{\mu_1 \mu_2}^{2} + \phi_{12} \sigma_{\mu_1}^{2} $, 
$ \delta_4 = \mathbb{E}[y_2 w_2] = \phi_{21} \sigma_{\mu_1 \mu_2}^{2} + \phi_{22} \sigma_{\mu_1}^{2}  $, 
$\delta_5 = \mathbb{E}[w_{1}^{2}] = \phi_{11}^{2} \sigma_{\mu_1}^{2} + \phi_{12}^{2} \sigma_{\mu_2}^{2} + \phi_{11} \phi_{12} \sigma_{\mu_1 \mu_2}^{2} $, 
$ \delta_6 = \mathbb{E}[w_{2}^{2}] = \phi_{21}^{2} \sigma_{\mu_1}^{2} + \phi_{22}^{2} \sigma_{\mu_2}^{2} + \phi_{21} \phi_{22} \sigma_{\mu_1 \mu_2}^{2} $, 
$ \delta_7 = \mathbb{E}[w_1 w_2] = \phi_{11} \phi_{21} \sigma_{\mu_1}^{2} + (\phi_{22} \phi_{11} + \phi_{12} \phi_{21}) \sigma_{\mu_1 \mu_2}^{2} + \phi_{22} \phi_{12} \sigma_{\mu_2}^{2} $,
$\delta_8 = \mathbb{E}[w_1 \xi_1] = \phi_{11}\sigma_{\mu_1}^{2} + \phi_{12} \sigma_{\mu_1 \mu_2}^{2} $,
$\delta_9 =\mathbb{E}[w_2 \xi_2] = \phi_{21}\sigma_{\mu_1 \mu_2}^{2} + \phi_{22} \sigma_{\mu_2}^{2} $.
Now rewriting team optimal cost function we obtain,
\begin{eqnarray*}
\mathcal{J}_{\TO_{ER}}^{*} = \min_{\alpha_{11},\alpha_{12},
\alpha_{21}, \alpha_{22}} 2 \alpha_{11}^{2} \sigma_{y_1}^{2} - 
2 \alpha_{11} \alpha_{21} \sigma_{y_1,y_2}^{2} + \alpha_{21}^{2} \sigma_{y_2}^{2} 
+ 2 \alpha_{11} \alpha_{12} \delta_1 - \alpha_{21} \alpha_{12} \delta_2 - \alpha_{11} \alpha_{22} \delta_{3} + \alpha_{22} \alpha_{21} \delta_4
+ \\ 
2 \alpha_{12}^{2} \delta_5 - 2 \alpha_{12}\alpha_{22} \delta_7 + \alpha_{22}^{2} \delta_6
+ 2( \alpha_{11} \sigma_{y_1}^{2} + \alpha_{21} \sigma_{y_2}^{2} )
+2( \alpha_{12} \delta_8 + \alpha_{22} \delta_9).
\end{eqnarray*}
Differentiating above expression with respect to $\alpha_{11}, \alpha_{12}, \alpha_{21}, \alpha_{22}$ and equating to 0. We have
\begin{eqnarray*}
\left[ \begin{array}{cccc}
4\sigma_{y_1}^{2} &  -2 \sigma_{y_1 y_2}^{2} & 2 \delta_1 & -\delta_3  \\
-2 \sigma_{y_1 y_2}^{2} & 2 \sigma_{y_1}^{2}  &  - \delta_2 & \delta_4 \\
2 \delta_1 & -\delta_2 & 4 \delta_5 & -2 \delta_7 \\
-\delta_3 & \delta_4 & -2 \delta_7 &  2 \delta_6
\end{array} \right]
\left[ \begin{array}{c}
\alpha_{11} \\
\alpha_{21} \\
\alpha_{12} \\
\alpha_{22}
\end{array} \right]
 = 
 \left[ \begin{array}{c}
-2 \sigma_{y_1}^{2} \\
-2 \sigma_{y_2}^{2} \\
-2 \delta_8 \\
-2 \delta_9
\end{array} \right]. 
\end{eqnarray*}
Notice that computing optimal $\alpha_{11}, \alpha_{12}, \alpha_{21}, \alpha_{22}$ via solving linear systems of equations and finding optimal expected cost is computationally tedious. 
Without loss of generality, we suppose $\sigma_{\mu_1}^{2} = \sigma_{\mu_2}^{2} = 1$ and $\sigma_{\mu_1 \mu_2}^{2} = \frac{1}{4}.$ Furthermore, we fix $\phi_{11} $, $\phi_{12}$, $\phi_{21} $, $\phi_{22}$ and evaluate the minimum team cost under optimal $\alpha_{11}^{*}, \alpha_{12}^{*}, \alpha_{21}^{*}, \alpha_{22}^{*}$. 
Note that $\phi_{11} $, $\phi_{12}$, $\phi_{21} $, $\phi_{22}$ determines the correlation of extra randomness with observations  available at decision makers.
From numerical computation in table (\ref{LQGstaticteam}), we make following concluding remarks.
\begin{enumerate}
\item In distributed static LQG team problem without extra randomness, the team optimal cost is highest.
\item  In distributed static LQG team problem, only one decision maker having extra randomness which is correlated with $\xi$ do not lead to improve in the team optimal cost. Instead it lead to increase in the team optimal cost.
\item In distributed static LQG team problem,all decision maker having extra randomness which is correlated with $\xi$ lead to improvement in  the team optimal cost. Thus we have strict inequality between $ \mathcal{J}_{\TO}^{*}$ and $\mathcal{J}_{\TO_{ER}}^{*}$, that means $\mathcal{J}_{\TO_{ER}}^{*} < \mathcal{J}_{\TO}^{*}$.
\item if an extra randomness provided by a consultant is an average of the observations $\mu_1$ and $\mu_2$, then  team optimal cost is best than any other convex combination of the observations $\mu_1$ and $\mu_2$. Hence correlation coefficient $\phi_{ij}$  for $1 \leq i,j \leq 2$ plays significant role to attain minimal team optimal cost in distributed static LQG team problem with extra randomness dependent on $\xi$.
\end{enumerate}
\begin{table}
\begin{center}
\begin{tabular}{ |l|l|l|l| }
\hline
 & $(\phi_{11},\phi_{12}, \phi_{21} , \phi_{22})$   & $ (\alpha_{11}^{*}, \alpha_{12}^{*}, \alpha_{21}^{*}, \alpha_{22}^{*})$ & $\min_{\alpha}\mathbb{E}[\kappa(\alpha, \xi) ] $  \\ \hline
No randomization & $(0,0,0,0)$ & $(-0.6452,-1.1613,0,0)$ & $-1.806$ \\ \hline
$DM_1$ have randomness  & $(\frac{1}{4}, \frac{3}{4}, 0, 0)$ & $(0, -1, -0.3024, 2.7513)$  & $-0.477$  \\   \hline
Both DM have randomness  & $(\frac{1}{2}, \frac{1}{2}, \frac{1}{2}, \frac{1}{2})$ & $(-0.3434, -0.7046, -2.7862, -4.0062)$ & $-5.2974$ \\
  \hline
Both DM have randomness  & $(\frac{2}{3}, \frac{1}{3}, \frac{3}{4}, \frac{1}{4})$ & $(-0.5122, -1.4833,-2.6067,-3.2171) $ & $-4.5211$ \\
  \hline
Both DM have randomness  & $(\frac{1}{3}, \frac{2}{3}, \frac{1}{4}, \frac{3}{4})$ & $ (-0.7045, -0.7058, -0.6765, -1.522)$ & $-3.6923$ \\
  \hline
\end{tabular}
\caption{Comparison of expected cost with different randomization provided to DM}
\label{LQGstaticteam}
\end{center}

\end{table}

%% file: discrete-zero-sum-teamgame.tex
\subsection{Proof of Lemma \ref{thm:two-teamzero-sumgame}}
\label{proof:two-teamzero-sumgame}
We prove our claim via illustrating an example of two-team discrete game.

Consider two-team label them as Team $1$ and $2$, Team $1$ consists of a decision maker and Team $2$ comprises  two decision makers. Let $\xi = [\mu_1, s_1, s_2]^{T}$ denote an environment or the state of nature; it is random vector with discrete distribution $p(\xi)$. 
Each decision maker observes an environment partially since decision maker are situated distributed manner. 

Let $y_1 = \eta(\xi)$ denote an observation available at decision maker of Team $1$;
$z_j = \zeta_j(\xi)$ represent an observation available at $\DM_j$ of Team $2$. Decision rule at Team $1$ and $2$
is 
\begin{equation*}
\gamma_{1} : y_1 \rightarrow u_1
\end{equation*}
and 
\begin{equation*}
\delta_j: z_j \rightarrow v_j
\end{equation*}
 $j = 1,2$. 

Without loss of generality, we assume that $\mu_1, s_j$ is  binary random variable take values $\{0,1\}$; $y_1 = \eta(\xi) = \mu_1$, $z_j = \zeta_j(\xi) = s_j$, for $j = 1,2$. Moreover, we consider $u_1, v_j \in \{L, R\}$,  $j = 1,2.$

Binary random variable $\mu_1$, $s_1$ and $s_2$ defined as follows.
\[
\mu_1 = \left\{
\begin{array}{l l}
1 & \quad \text{with prob. } p_1 \\
0 & \quad \text{with prob.} 1 - p_1.
\end{array}
\right.
\]
\[
s_1 = \left\{
\begin{array}{l l}
\mu_1 & \quad \text{with prob. } p \\
0 & \quad \text{with prob.} 1 - p.
\end{array}
\right.
\]
\[
s_2 = \left\{
\begin{array}{l l}
1- \mu_1 & \quad \text{with prob. } q \\
s_1 & \quad \text{with prob.} 1 - q.
\end{array}
\right.
\]
The joint distribution of $(\mu_1, s_1, s_2)$ is  $\mathbb{P}(\mu_1, s_1, s_2)$  and  is written as
\begin{eqnarray*}
\mathbb{P}(\mu_1 = 0, s_1 =0, s_2 = 0) = (1 - p_1)(1 - q)
\end{eqnarray*}
\begin{eqnarray*}
\mathbb{P}(\mu_1 = 0, s_1 =0, s_2 = 1) = (1 - p_1)q
\end{eqnarray*}
\begin{eqnarray*}
\mathbb{P}(\mu_1 = 0, s_1 =1, s_2 = 0) = 0
\end{eqnarray*}
\begin{eqnarray*}
\mathbb{P}(\mu_1 = 0, s_1 = 1, s_2 = 1) = 0
\end{eqnarray*}
\begin{eqnarray*}
\mathbb{P}(\mu_1 = 1, s_1 =0, s_2 = 0) = p_1 (1 -p)
\end{eqnarray*}
\begin{eqnarray*}
\mathbb{P}(\mu_1 = 1, s_1 =0, s_2 = 1) = 0
\end{eqnarray*}
\begin{eqnarray*}
\mathbb{P}(\mu_1 = 0, s_1 =1, s_2 = 0) = p_1 p q
\end{eqnarray*}
\begin{eqnarray*}
\mathbb{P}(\mu_1 = 1, s_1 =1, s_2 = 1) = p_1p(1-q)
\end{eqnarray*}

There are four possible decision rule available at each decision maker.
The decision rule of a decision maker is 
\[
u_1^{1} = \gamma_{1}^{1}(y_{1}) = \gamma_{1}^{1}(\mu_1) = 
 \left\{
\begin{array}{l l}
L & \quad \text{if } \mu_1 = 0 \\
R & \quad \text{if } \mu_1 = 1
\end{array}
\right.
\]
\[
u_1^{2} = \gamma_{1}^{2}(y_{1}) = \gamma_{1}^{2}(\mu_1) = 
 \left\{
\begin{array}{l l}
L & \quad \text{if } \mu_1 = 1 \\
R & \quad \text{if } \mu_1 = 0
\end{array}
\right.
\]
\[
u_1^{3} = \gamma_{1}^{3}(y_{1}) = \gamma_{1}^{3}(\mu_1) = 
 \left\{
\begin{array}{l l}
L & \quad \text{if } \mu_1 = 1 \text{ or } \mu_1 = 0
\end{array}
\right.
\]
\[
u_1^{4} = \gamma_{1}^{4}(y_{1}) = \gamma_{1}^{4}(\mu_1) = 
 \left\{
\begin{array}{l l}
R & \quad \text{if } \mu_1 = 1 \text{ or } \mu_1 = 0
\end{array}
\right.
\]
\[
v_1^{1} = \delta_{1}^{1}(z_{1}) = \delta_{1}^{1}(s_1) = 
 \left\{
\begin{array}{l l}
L & \quad \text{if } s_1 = 0 \\
R & \quad \text{if } s_1 = 1
\end{array}
\right.
\]
\[
v_{1}^{2} = \delta_{1}^{2}(z_{1}) = \delta_{1}^{2}(s_1) = 
 \left\{
\begin{array}{l l}
L & \quad \text{if } s_1 = 1 \\
R & \quad \text{if } s_1 = 0
\end{array}
\right.
\]
\[
v_{1}^{3} = \delta_{1}^{3}(z_{1}) = \delta_{1}^{3}(s_1) = 
 \left\{
\begin{array}{l l}
L & \quad \text{if } s_1 = 1 \text{ or } s_1 = 0
\end{array}
\right.
\]
\[
v_{1}^{4} =\delta_{1}^{4}(z_{1}) = \delta_{1}^{4}(s_1) = 
 \left\{
\begin{array}{l l}
R & \quad \text{if } s_1 = 1 \text{ or } s_1 = 0
\end{array}
\right.
\]
\[
v_{2}^{1} = \delta_{2}^{1}(z_{2}) = \delta_{2}^{1}(s_2) = 
 \left\{
\begin{array}{l l}
L & \quad \text{if } s_2 = 0 \\
R & \quad \text{if } s_2 = 1
\end{array}
\right.
\]
\[
v_{2}^{2} = \delta_{2}^{2}(z_{2}) = \delta_{2}^{2}(s_2) = 
 \left\{
\begin{array}{l l}
L & \quad \text{if } s_2 = 1 \\
R & \quad \text{if } s_2 = 0
\end{array}
\right.
\]
\[
v_{2}^{3}\delta_{2}^{3}(z_{2}) = \delta_{2}^{3}(s_2) = 
 \left\{
\begin{array}{l l}
L & \quad \text{if } s_2 = 1 \text{ or } s_2 = 0
\end{array}
\right.
\]
\[
v_{2}^{4} = \delta_{2}^{4}(z_{2}) = \delta_{2}^{4}(s_2) = 
 \left\{
\begin{array}{l l}
R & \quad \text{if } s_2 = 1 \text{ or } s_2 = 0
\end{array}
\right.
\]
We next formulate team vs team zero-sum game, Team $1$ seeks to maximize the expected
payoff whereas Team $2$  seeks to minimize the expected payoff. 
We  describe payoff matrix in table \ref{twpteampayoff} .
\begin{table}
\begin{center}
\begin{tabular}{ |l|l|l|l|l| }
  \hline
  & LL & LR & RL & RR \\ \hline
  L & 20 & 0 & 1 & 30 \\ \hline
  R & 20 & 1 & 0 & 30 \\
  \hline
\end{tabular}
\end{center}
\caption{Payoff matrix: team vs team zero-sum game  }
\label{twpteampayoff}
\end{table}
In table \ref{twpteampayoff} row vector denotes actions of Team $1$ and corresponding payoff;   column vector denotes actions of Team $2$ and corresponding payoff.
Since observations available at each decision maker in team is function of state of nature $\xi$ and $\xi$ is random variable, we evaluate the expected payoff for different actions of decision makers and it is
\[
\mathbb{E}\left[\kappa\left(  \gamma_{1}^{l}(\mu_1), \delta_{1}^{m}(s_1), \delta_{2}^{n}(s_2) \right)\right]
 =  \sum_{\mu_1,s_1,s_2  \in \{ 0,1\}^{3}}\kappa\left(  \gamma_{1}^{l}(\mu_1),\delta_{1}^{m}(s_1)\delta_{2}^{n}(s_2)\right) \mathbf{P}(\mu_1,s_1,s_2).
\]
where $1 \leq l,m,n \leq 4$.
Enumerating  the expected payoff over all possible actions of decision makers, we obtain
\[
\mathbb{E}\left[\kappa\left(  \gamma_{1}^{1}(\mu_1),\delta_{1}^{1}(s_1)\delta_{2}^{1}(s_2) \right)\right] = 20 - 20q +20p_1q+10p_1p-30p_1pq
\]
\[
\mathbb{E}\left[\kappa\left(  \gamma_{1}^{2}(\mu_1),\delta_{1}^{1}(s_1)\delta_{2}^{1}(s_2) \right)\right] =40-40p_1-19q+19p_1q-29p_1pq+30p_1p
\]
\[
\mathbb{E}\left[\kappa\left(  \gamma_{1}^{3}(\mu_1),\delta_{1}^{1}(s_1)\delta_{2}^{1}(s_2) \right)\right] = 20-20q+30p_1p-29p_1pq
\]
\[
\mathbb{E}\left[\kappa\left(  \gamma_{1}^{4}(\mu_1),\delta_{1}^{1}(s_1)\delta_{2}^{1}(s_2) \right)\right] = 20-19q+19p_1q+10p_1p-30p_1pq
\]
\[
\mathbb{E}\left[\kappa\left(  \gamma_{1}^{1}(\mu_1),\delta_{1}^{2}(s_1)\delta_{2}^{1}(s_2) \right)\right] = 1-p_1+29q-29p_1q+19p_1pq+p_1p
\]
\[
\mathbb{E}\left[\kappa\left(  \gamma_{1}^{2}(\mu_1),\delta_{1}^{2}(s_1)\delta_{2}^{1}(s_2) \right)\right] =  30-31p_1q+p_1+20p_1pq
\]
\[
\mathbb{E}\left[\kappa\left(  \gamma_{1}^{3}(\mu_1),\delta_{1}^{2}(s_1)\delta_{2}^{1}(s_2) \right)\right] = 1+29q-30p_1q+20p_1pq
\]
\[
\mathbb{E}\left[\kappa\left(  \gamma_{1}^{4}(\mu_1),\delta_{1}^{2}(s_1)\delta_{2}^{1}(s_2) \right)\right] = 30q -30p_1q +19p_1pq + p_1p
\]
\[
\mathbb{E}\left[\kappa\left(  \gamma_{1}^{1}(\mu_1),\delta_{1}^{3}(s_1)\delta_{2}^{1}(s_2) \right)\right] = 20 -20q +19p_1pq + p_1p
\]
\[
\mathbb{E}\left[\kappa\left(  \gamma_{1}^{2}(\mu_1),\delta_{1}^{3}(s_1)\delta_{2}^{1}(s_2) \right)\right] = 20-19q-p_1q+20p_1pq
\]
\[
\mathbb{E}\left[\kappa\left(  \gamma_{1}^{3}(\mu_1),\delta_{1}^{3}(s_1)\delta_{2}^{1}(s_2) \right)\right] =  20 -20q + 20p_1pq
\]
\[
\mathbb{E}\left[\kappa\left(  \gamma_{1}^{4}(\mu_1),\delta_{1}^{3}(s_1)\delta_{2}^{1}(s_2) \right)\right] = 20-19q -p_1q +19p_1pq + p_1p
\]
\[
\mathbb{E}\left[\kappa\left(  \gamma_{1}^{1}(\mu_1),\delta_{1}^{4}(s_1)\delta_{2}^{1}(s_2) \right)\right] = 1 -p_1+29q-29p_1q+30p_1p-30p_1pq
\]
\[
\mathbb{E}\left[\kappa\left(  \gamma_{1}^{2}(\mu_1),\delta_{1}^{4}(s_1)\delta_{2}^{1}(s_2) \right)\right] = 30q-31p_1q+p_1-29p_1pq+30p_1p
\]
\[
\mathbb{E}\left[\kappa\left(  \gamma_{1}^{3}(\mu_1),\delta_{1}^{4}(s_1)\delta_{2}^{1}(s_2) \right)\right] =  1 + 29q-29p_1q+29p_1p-29p_1pq
\]
\[
\mathbb{E}\left[\kappa\left(  \gamma_{1}^{4}(\mu_1),\delta_{1}^{4}(s_1)\delta_{2}^{1}(s_2) \right)\right] = 30q-30p_1q+30p_1p-30p_1pq
\]
\[
\mathbb{E}\left[\kappa\left(  \gamma_{1}^{1}(\mu_1),\delta_{1}^{1}(s_1)\delta_{2}^{2}(s_2) \right)\right] = 20q-20p_1q+p_1+30p_1pq
\]
\[
\mathbb{E}\left[\kappa\left(  \gamma_{1}^{2}(\mu_1),\delta_{1}^{1}(s_1)\delta_{2}^{2}(s_2) \right)\right] = 1-p_1+19q-19p_1q+29p_1pq+p_1p
\]
\[
\mathbb{E}\left[\kappa\left(  \gamma_{1}^{3}(\mu_1),\delta_{1}^{1}(s_1)\delta_{2}^{2}(s_2) \right)\right] =  20q -20p_1q +29p_1pq+p_1p
\]
\[
\mathbb{E}\left[\kappa\left(  \gamma_{1}^{4}(\mu_1),\delta_{1}^{1}(s_1)\delta_{2}^{2}(s_2) \right)\right] = 1+ 19q-19p_1q-p_1p+30p_1pq
\]
\[
\mathbb{E}\left[\kappa\left(  \gamma_{1}^{1}(\mu_1),\delta_{1}^{2}(s_1)\delta_{2}^{2}(s_2) \right)\right] = 30 - 29q + 29p_1q -10 p_1p -19p_1pq
\]
\[
\mathbb{E}\left[\kappa\left(  \gamma_{1}^{2}(\mu_1),\delta_{1}^{2}(s_1)\delta_{2}^{2}(s_2) \right)\right] = 30-10p_1-30q+30p_1q-20p_1pq
\]
\[
\mathbb{E}\left[\kappa\left(  \gamma_{1}^{3}(\mu_1),\delta_{1}^{2}(s_1)\delta_{2}^{2}(s_2) \right)\right] =  30- 29q+29p_1q-10p_1p-20p_1pq
\]
\[
\mathbb{E}\left[\kappa\left(  \gamma_{1}^{4}(\mu_1),\delta_{1}^{2}(s_1)\delta_{2}^{2}(s_2) \right)\right] = 30-30q+30p_1q-10p_1p-19p_1pq
\]
\[
\mathbb{E}\left[\kappa\left(  \gamma_{1}^{1}(\mu_1),\delta_{1}^{3}(s_1)\delta_{2}^{2}(s_2) \right)\right] = 20q-20p_1q+p_1+19p_1p-19p_1pq
\]
\[
\mathbb{E}\left[\kappa\left(  \gamma_{1}^{2}(\mu_1),\delta_{1}^{3}(s_1)\delta_{2}^{2}(s_2) \right)\right] = 1-p_1+19q-19p_1q+20p_1p-20p_1pq
\]
\[
\mathbb{E}\left[\kappa\left(  \gamma_{1}^{3}(\mu_1),\delta_{1}^{3}(s_1)\delta_{2}^{2}(s_2) \right)\right] =  20q-20p_1q + 20p_1pq
\]
\[
\mathbb{E}\left[\kappa\left(  \gamma_{1}^{4}(\mu_1),\delta_{1}^{3}(s_1)\delta_{2}^{2}(s_2) \right)\right] = 1 + 19q-19p_1q+19p_1p-19p_1pq
\]
\[
\mathbb{E}\left[\kappa\left(  \gamma_{1}^{1}(\mu_1),\delta_{1}^{4}(s_1)\delta_{2}^{2}(s_2) \right)\right] = 30 -29q+29p_1q-29p_1p+29p_1pq
\]
\[
\mathbb{E}\left[\kappa\left(  \gamma_{1}^{2}(\mu_1),\delta_{1}^{4}(s_1)\delta_{2}^{2}(s_2) \right)\right] = 30-30q+30p_1q-29p_1p+29p_1pq
\]
\[
\mathbb{E}\left[\kappa\left(  \gamma_{1}^{3}(\mu_1),\delta_{1}^{4}(s_1)\delta_{2}^{2}(s_2) \right)\right] = 30-29q+29p_1q-29p_1p+29p_1pq
\]
\[
\mathbb{E}\left[\kappa\left(  \gamma_{1}^{4}(\mu_1),\delta_{1}^{4}(s_1)\delta_{2}^{2}(s_2) \right)\right] = 30-30q+30p_1q-30p_1p+30p_1pq
\]
\[
\mathbb{E}\left[\kappa\left(  \gamma_{1}^{1}(\mu_1),\delta_{1}^{1}(s_1)\delta_{2}^{3}(s_2) \right)\right] = 20 - 20 p_1p
\]
\[
\mathbb{E}\left[\kappa\left(  \gamma_{1}^{2}(\mu_1),\delta_{1}^{1}(s_1)\delta_{2}^{3}(s_2) \right)\right] = 20 - 19p_1 p
\]
\[
\mathbb{E}\left[\kappa\left(  \gamma_{1}^{3}(\mu_1),\delta_{1}^{1}(s_1)\delta_{2}^{3}(s_2) \right)\right] = 20 - 19p_1p
\]
\[
\mathbb{E}\left[\kappa\left(  \gamma_{1}^{4}(\mu_1),\delta_{1}^{1}(s_1)\delta_{2}^{3}(s_2) \right)\right] = 20 - 20 p_1 p
\]
\[
\mathbb{E}\left[\kappa\left(  \gamma_{1}^{1}(\mu_1),\delta_{1}^{2}(s_1)\delta_{2}^{3}(s_2) \right)\right] = 1 - p_1 +20 p_1 p
\]
\[
\mathbb{E}\left[\kappa\left(  \gamma_{1}^{2}(\mu_1),\delta_{1}^{2}(s_1)\delta_{2}^{3}(s_2) \right)\right] = p_1 + 19 p_1 p
\]
\[
\mathbb{E}\left[\kappa\left(  \gamma_{1}^{3}(\mu_1),\delta_{1}^{2}(s_1)\delta_{2}^{3}(s_2) \right)\right] = p_1 + 19 p_1p
\]
\[
\mathbb{E}\left[\kappa\left(  \gamma_{1}^{4}(\mu_1),\delta_{1}^{2}(s_1)\delta_{2}^{3}(s_2) \right)\right] = 20p_1 p
\]
\[
\mathbb{E}\left[\kappa\left(  \gamma_{1}^{1}(\mu_1),\delta_{1}^{3}(s_1)\delta_{2}^{3}(s_2) \right)\right] = 20
\]
\[
\mathbb{E}\left[\kappa\left(  \gamma_{1}^{2}(\mu_1),\delta_{1}^{3}(s_1)\delta_{2}^{3}(s_2) \right)\right] = 20
\]
\[
\mathbb{E}\left[\kappa\left(  \gamma_{1}^{3}(\mu_1),\delta_{1}^{3}(s_1)\delta_{2}^{3}(s_2) \right)\right] = 20
\]
\[
\mathbb{E}\left[\kappa\left(  \gamma_{1}^{4}(\mu_1),\delta_{1}^{3}(s_1)\delta_{2}^{3}(s_2) \right)\right] = 20
\]
\[
\mathbb{E}\left[\kappa\left(  \gamma_{1}^{1}(\mu_1),\delta_{1}^{4}(s_1)\delta_{2}^{3}(s_2) \right)\right] = 1 - p_1 
\]
\[
\mathbb{E}\left[\kappa\left(  \gamma_{1}^{2}(\mu_1),\delta_{1}^{4}(s_1)\delta_{2}^{3}(s_2) \right)\right] = p_1 
\]
\[
\mathbb{E}\left[\kappa\left(  \gamma_{1}^{3}(\mu_1),\delta_{1}^{4}(s_1)\delta_{2}^{3}(s_2) \right)\right] = 1
\]
\[
\mathbb{E}\left[\kappa\left(  \gamma_{1}^{4}(\mu_1),\delta_{1}^{4}(s_1)\delta_{2}^{3}(s_2) \right)\right] = 0
\]
\[
\mathbb{E}\left[\kappa\left(  \gamma_{1}^{1}(\mu_1),\delta_{1}^{1}(s_1)\delta_{2}^{4}(s_2) \right)\right] = p_1 + 29p_1p
\]
\[
\mathbb{E}\left[\kappa\left(  \gamma_{1}^{2}(\mu_1),\delta_{1}^{1}(s_1)\delta_{2}^{4}(s_2) \right)\right] = 1 -p_1  +20p_1p
\]
\[
\mathbb{E}\left[\kappa\left(  \gamma_{1}^{3}(\mu_1),\delta_{1}^{1}(s_1)\delta_{2}^{4}(s_2) \right)\right] = 30p_1p
\]
\[
\mathbb{E}\left[\kappa\left(  \gamma_{1}^{4}(\mu_1),\delta_{1}^{1}(s_1)\delta_{2}^{4}(s_2) \right)\right] = 1+ 29 p_1 p 
\]
\[
\mathbb{E}\left[\kappa\left(  \gamma_{1}^{1}(\mu_1),\delta_{1}^{2}(s_1)\delta_{2}^{4}(s_2) \right)\right] = 30p_1 -29p_1p
\]
\[
\mathbb{E}\left[\kappa\left(  \gamma_{1}^{2}(\mu_1),\delta_{1}^{2}(s_1)\delta_{2}^{4}(s_2) \right)\right] = 30 - 30 p_1 p
\]
\[
\mathbb{E}\left[\kappa\left(  \gamma_{1}^{3}(\mu_1),\delta_{1}^{2}(s_1)\delta_{2}^{4}(s_2) \right)\right] = 30 - 30 p_1 p
\]
\[
\mathbb{E}\left[\kappa\left(  \gamma_{1}^{4}(\mu_1),\delta_{1}^{2}(s_1)\delta_{2}^{4}(s_2) \right)\right] =  30 -29 p_1 p
\]
\[
\mathbb{E}\left[\kappa\left(  \gamma_{1}^{1}(\mu_1),\delta_{1}^{3}(s_1)\delta_{2}^{4}(s_2) \right)\right] =  p_1
\]
\[
\mathbb{E}\left[\kappa\left(  \gamma_{1}^{2}(\mu_1),\delta_{1}^{3}(s_1)\delta_{2}^{4}(s_2) \right)\right] = 1 - p_1
\]
\[
\mathbb{E}\left[\kappa\left(  \gamma_{1}^{3}(\mu_1),\delta_{1}^{3}(s_1)\delta_{2}^{4}(s_2) \right)\right] = 0
\]
\[
\mathbb{E}\left[\kappa\left(  \gamma_{1}^{4}(\mu_1),\delta_{1}^{3}(s_1)\delta_{2}^{4}(s_2) \right)\right] = 1
\]
\[
\mathbb{E}\left[\kappa\left(  \gamma_{1}^{1}(\mu_1),\delta_{1}^{4}(s_1)\delta_{2}^{4}(s_2) \right)\right] = 30
\]
\[
\mathbb{E}\left[\kappa\left(  \gamma_{1}^{2}(\mu_1),\delta_{1}^{4}(s_1)\delta_{2}^{4}(s_2) \right)\right] = 30
\]
\[
\mathbb{E}\left[\kappa\left(  \gamma_{1}^{3}(\mu_1),\delta_{1}^{4}(s_1)\delta_{2}^{4}(s_2) \right)\right] = 30
\]
\[
\mathbb{E}\left[\kappa\left(  \gamma_{1}^{4}(\mu_1),\delta_{1}^{4}(s_1)\delta_{2}^{4}(s_2) \right)\right] = 30
\]
From above expression, it is difficult to make any comment on saddle point solution of zero-sum game. Thus we suppose $p_1 = \frac{1}{4}$, $p = \frac{1}{3}$ and $q = \frac{2}{3}$ but it is also possible that under different range of $p_1, p, q$ our claim  holds true. Rewriting expected payoff matrix for zero-sum game in table \ref{twoteamexpectpay1}.
\begin{table}
\begin{center}
\begin{tabular}{ |l|l|l|l|l| }
  \hline
  & $\gamma_{1}^{1}(\mu_1)$ & $\gamma_{1}^{2}(\mu_1)$ & $\gamma_{1}^{3}(\mu_1)$ & $\gamma_{1}^{4}(\mu_1)$ \\ \hline
$\delta_{1}^{1}(s_1)\delta_{2}^{1}(s_2)$ & 9.16 & 22.3 & 7.54 & 9.66 \\ \hline
$\delta_{1}^{2}(s_1)\delta_{2}^{1}(s_2)$ & 16.39 & 26.18 & 16.45 & 16.13 \\   \hline
$\delta_{1}^{3}(s_1)\delta_{2}^{1}(s_2)$ & 7.80 & 8.27 & 7.77 & 8.30 \\  \hline
$\delta_{1}^{4}(s_1)\delta_{2}^{1}(s_2)$ & 16.08 & 15.72 & 16.30 & 15.83 \\  \hline
$\delta_{1}^{1}(s_1)\delta_{2}^{2}(s_2)$ & 11.91 & 11.94 & 11.69 & 12.08 \\ \hline
$\delta_{1}^{2}(s_1)\delta_{2}^{2}(s_2)$ & 13.61 & 11.38 & 13.55 & 13.11 \\   \hline
$\delta_{1}^{3}(s_1)\delta_{2}^{2}(s_2)$ & 10.77 & 10.80 & 11.11 & 11.02 \\  \hline
$\delta_{1}^{4}(s_1)\delta_{2}^{2}(s_2)$ & 14.69 & 14.19 & 14.69 &  14.16 \\  \hline
$\delta_{1}^{1}(s_1)\delta_{2}^{3}(s_2)$ & 18.33 & 18.41 & 18.41 & 18.33 \\ \hline
$\delta_{1}^{2}(s_1)\delta_{2}^{3}(s_2)$ & 2.41 & 1.83 & 1.83 & 1.66 \\   \hline
$\delta_{1}^{3}(s_1)\delta_{2}^{3}(s_2)$ & 20  & 20  & 20  & 20 \\  \hline
$\delta_{1}^{4}(s_1)\delta_{2}^{3}(s_2)$ & 0.75 & 0.25 & 1 & 0 \\  \hline
$\delta_{1}^{1}(s_1)\delta_{2}^{4}(s_2)$ & 2.66 & 2.41 & 2.5 & 3.41 \\ \hline
$\delta_{1}^{2}(s_1)\delta_{2}^{4}(s_2)$ & 5.08 & 27.5 & 27.5 & 27.58 \\   \hline
$\delta_{1}^{3}(s_1)\delta_{2}^{4}(s_2)$ & 0.25 & 0.75 & 0 & 1 \\  \hline
$\delta_{1}^{4}(s_1)\delta_{2}^{4}(s_2)$ & 30 & 30  & 30  & 30 \\  \hline
\end{tabular}
\end{center}
\caption{Two-team zero-sum game with expected payoff  matrix}
\label{twoteamexpectpay1}
\end{table}
In table \ref{twoteamexpectpay1}, row vector denote strategies of a Team $2$, column vector denote strategies of a Team $1$ and corresponding expected payoff. 
Here, Team $2$ wishes to minimize the expected payoff and  Team $1$ wishes to maximize the 
expected payoff. 
The security level of Team $1$ is 
\[
\underline{V}(A) =  \max_{j} \min_{i} a_{ij}  = 0.25
\]
Similarly, the security level of Team $2$ is 
\[
\overline{V}(A) = \min_{i} \max_{j} a_{ij}  = 1.
\]

Notice that we have $ \overline{V}(A) > \underline{V}(A) $, it implies this game do not admit the pure strategy saddle point solution.

\subsubsection{Role of the private randomness independent of $\xi$}

We are interested to understand the role of the private randomness in two-team zero-sum game.  We assume a coordinator provides the private randomness to decision maker of a team, say
Team $1$ decision maker. Further we assume that these private randomization is independent of $\xi$.

Consider Team $1$ decision maker has private randomization over its strategies
and  plays strategy $\gamma_{1}^{i}(\mu_1)$ with probability $a_i$ for $1 \leq  i \leq 4$ and $\sum_{i=1}^{4} a_i = 1.$ That is
\[
\gamma_{1}(\mu_1) = \left\{
\begin{array}{l l}
\gamma_{1}^{1}(\mu_1) & \quad \text{with prob. } a_1 \\
\gamma_{1}^{2}(\mu_1) & \quad \text{with prob. } a_2 \\
\gamma_{1}^{3}(\mu_1) & \quad \text{with prob. } a_3 \\
\gamma_{1}^{4}(\mu_1) & \quad \text{with prob. } a_4. \\
\end{array}
\right.
\]

Then the expected payoff is
\[
\mathbb{E}\left[\kappa\left( \gamma_{1}(\mu_1) \delta_{1}^{j}(s_1) \delta_{2}^{k}(s_2) \right)\right] = \sum_{i = 1}^{4} \mathbb{E}\left[\kappa\left( (\gamma_{1}(\mu_1) = \gamma_{1}^{i}(\mu_1) ) \delta_{1}^{j}(s_1) \delta_{2}^{k}(s_2) \right)\right] a_i
\]
for $ 1 \leq j,k \leq 4$. We have evaluated the expected payoff and given in table \ref{twoteamexpectprivrand}.
\begin{table}
\begin{center}
\begin{tabular}{ |l|l| }
  \hline
$\delta_{1}^{1}(s_1)\delta_{2}^{1}(s_2)$ & $9.16a_1 + 22.3a_2 + 7.54a_3 + 9.66a_4 $\\ \hline
$\delta_{1}^{2}(s_1)\delta_{2}^{1}(s_2)$ & $16.39a_1 + 26.18a_2 + 16.45a_3 + 16.13 a_4$ \\   \hline
$\delta_{1}^{3}(s_1)\delta_{2}^{1}(s_2)$ & $7.80a_1 + 8.27a_2 +7.77a_3 + 8.30a_4$ \\  \hline
$\delta_{1}^{4}(s_1)\delta_{2}^{1}(s_2)$ & $16.08a_1 +15.72 a_2 + 16.30a_3 + 15.83a_4$ \\  \hline
$\delta_{1}^{1}(s_1)\delta_{2}^{2}(s_2)$ & $11.91a_1 + 11.94a_2 + 11.69a_3 + 12.08a_4$ \\ \hline
$\delta_{1}^{2}(s_1)\delta_{2}^{2}(s_2)$ & $13.61a_1 + 11.38a_2 + 13.55a_3 + 13.11 a_4$ \\   \hline
$\delta_{1}^{3}(s_1)\delta_{2}^{2}(s_2)$ & $10.77a_1 + 10.80a_2 + 11.11a_3 + 11.02 a_4$ \\  \hline
$\delta_{1}^{4}(s_1)\delta_{2}^{2}(s_2)$ & $14.69 a_1+ 14.19a_2 +14.69a_3 + 14.16a_4$ \\  \hline
$\delta_{1}^{1}(s_1)\delta_{2}^{3}(s_2)$ & $18.33a_1 + 18.41a_2 + 18.41a_3 + 18.33a_4$ \\ \hline
$\delta_{1}^{2}(s_1)\delta_{2}^{3}(s_2)$ & $2.41a_1 + 1.83a_2 + 1.83a_3 + 1.66 a_4$ \\   \hline
$\delta_{1}^{3}(s_1)\delta_{2}^{3}(s_2)$ & $20$\\  \hline
$\delta_{1}^{4}(s_1)\delta_{2}^{3}(s_2)$ & $0.75a_1 + 0.25a_2 + 1a_3$  \\  \hline
$\delta_{1}^{1}(s_1)\delta_{2}^{4}(s_2)$ & $2.66a_1 + 2.41a_2 + 2.5a_3 + 3.41a_4 $\\ \hline
$\delta_{1}^{2}(s_1)\delta_{2}^{4}(s_2)$ & $5.08a_1 + 27.5 a_2 + 27.5a_3 + 27.58a_4$ \\   \hline
$\delta_{1}^{3}(s_1)\delta_{2}^{4}(s_2)$ & $0.25a_1 + 0.75a_2+ 1a_4 $\\  \hline
$\delta_{1}^{4}(s_1)\delta_{2}^{4}(s_2)$ & $30$ \\  \hline
\end{tabular}
\end{center}
\caption{Two-team zero-sum game expected payoff with team $1$ has private randomization}
\label{twoteamexpectprivrand}
\end{table}
From table \ref{twoteamexpectprivrand}, notice that a Team $2$ best response will be $\left(\delta_{1}^{3}(s_1)\delta_{2}^{4}(s_2)\right)$ or $\left(\delta_{1}^{4}(s_1)\delta_{2}^{3}(s_2)\right)$ depend on probability vector $a= [a_1,a_2,a_3,a_4]$ at Team $1$ (i.e. private randomization). 
Without loss of generality, we assume $a_3 = a_4$, now observe that $a_1$ and $a_2$ determines the best response of Team $2$. We demonstrate this as follows.
\begin{enumerate}
\item If $a_1 < a_2$, thw  best response of team $2$ will be $\left(\delta_{1}^{3}(s_1)\delta_{2}^{4}(s_2)\right)$ and expected payoff  will be $(0.75a_1 + 0.25a_2 + 1a_3)$.
Further assume $a_2 = 2 a_1$, $a_3 = a_4 = \frac{1}{12}$, then $a_1= \frac{5}{18}$ and  
expected payoff is $0.43$.
\item  If $a_1 > a_2$, the  best response of team $2$ will be $\left(\delta_{1}^{4}(s_1)\delta_{2}^{3}(s_2)\right)$ and expected payoff  will be $(0.25a_1 + 0.75a_2+ 1a_4)$.
Similarly, we assume $a_1 = 2 a_2$, $a_3 = a_4 = \frac{1}{12}$, then $a_2= \frac{5}{18}$ and  
expected payoff is $0.43$.
\item  If $a_1 = a_2$, the  best response of team $2$ will be $\left(\delta_{1}^{3}(s_1)\delta_{2}^{4}(s_2)\right)$  or $\left(\delta_{1}^{4}(s_1)\delta_{2}^{3}(s_2)\right)$ and expected payoff  will be $a_1 + a_3$. We assume $a_3 = a_4 =  \frac{1}{12}$,then $a_2 = a_1 = \frac{5}{12}$ and  expected payoff is $0.5$.
\end{enumerate}

This implies that under private randomization at one of team, it do not admit Nash equilibrium solution.

Observe that the expected payoff of Team $2$ has improved from $1$ to $0.43$ if $a_1< a_2$ or $a_1>a_2$ and $0.5$ if $a_1 = a_2$ where Team $2$ wishes to minimize the expected payoff.


\begin{table}
\begin{center}
\begin{tabular}{ |l|l| }
  \hline
$\delta_{1}^{1}(s_1)\delta_{2}^{1}(s_2)$ & $16.33 $\\ \hline
$\delta_{1}^{2}(s_1)\delta_{2}^{1}(s_2)$ & $ 21.81$ \\   \hline
$\delta_{1}^{3}(s_1)\delta_{2}^{1}(s_2)$ & $ 8.10$ \\  \hline
$\delta_{1}^{4}(s_1)\delta_{2}^{1}(s_2)$ & $ 15.87$ \\  \hline
$\delta_{1}^{1}(s_1)\delta_{2}^{2}(s_2)$ & $ 11.92$ \\ \hline
$\delta_{1}^{2}(s_1)\delta_{2}^{2}(s_2)$ & $ 12.32$ \\   \hline
$\delta_{1}^{3}(s_1)\delta_{2}^{2}(s_2)$ & $ 10.83$ \\  \hline
$\delta_{1}^{4}(s_1)\delta_{2}^{2}(s_2)$ & $ 14.36$ \\  \hline
$\delta_{1}^{1}(s_1)\delta_{2}^{3}(s_2)$ & $ 18.38$ \\ \hline
$\delta_{1}^{2}(s_1)\delta_{2}^{3}(s_2)$ & $ 1.97$ \\   \hline
$\delta_{1}^{3}(s_1)\delta_{2}^{3}(s_2)$ & $ 20$\\  \hline
$\delta_{1}^{4}(s_1)\delta_{2}^{3}(s_2)$ & $ 0.43$  \\  \hline
$\delta_{1}^{1}(s_1)\delta_{2}^{4}(s_2)$ & $ 2.57$\\ \hline
$\delta_{1}^{2}(s_1)\delta_{2}^{4}(s_2)$ & $ 21.27$ \\   \hline
$\delta_{1}^{3}(s_1)\delta_{2}^{4}(s_2)$ & $ 0.57$\\  \hline
$\delta_{1}^{4}(s_1)\delta_{2}^{4}(s_2)$ & $ 30 $ \\  \hline
\end{tabular}
\end{center}
\caption{Two-team zero-sum game with team $1$ private randomization over its strategies $a_3 = a_4 = \frac{1}{12}$, $a_1= \frac{5}{18}$ $a_2= \frac{10}{18}$. }
\label{twoteamprivatrand2}
\end{table}
Now from table \ref{twoteamprivatrand2}, note that the best strategy of $DM_1$ and $DM_2$ in Team $2$  would be to play pure strategy as $\delta_{1}^{4}(s_1)$ and $\delta_{2}^{3}(s_2)$ to minimize the expected  payoff.

Furthermore, one of DM in Team $2$ having private randomness may not lead to improve in the expected  payoff.
To demonstrate this, consider $DM_1$ in Team $2$ has private randomization over his strategies.
\[
\delta_{1}(s_1) = \left\{
\begin{array}{l l}
\delta_{1}^{1}(s_1) & \quad \text{with prob. } b_1 \\
\delta_{1}^{2}(s_1) & \quad \text{with prob. } b_2 \\
\delta_{1}^{3}(s_1) & \quad \text{with prob. } b_3 \\
\delta_{1}^{4}(s_1) & \quad \text{with prob. } b_4 \\
\end{array}
\right.
\]
and $0 \leq b_i \leq 1$ for $1 \leq i \leq 4$, $\sum_{i=1}^{4} b_i = 1$.

 The expected payoff payoff is
 \[
\mathbb{E}\left[\kappa\left( \gamma_{1}(\mu_1) 
\delta_{1}(s_1) \delta_{2}^{k}(s_2) \right)\right] =
 \sum_{i,j}^{4} \mathbb{E}\left[\kappa\left( (\gamma_{1}
(\mu_1) = \gamma_{1}^{i}(\mu_1) ) (\delta_{1}(s_1) 
= \delta_{1}^{j}(s_1)) \delta_{2}^{k}(s_2) \right)\right] a_i b_j,
\]
\begin{table}
\begin{center}
\begin{tabular}{ |l|l| }
  \hline
$\delta_{2}^{1}(s_2)$ & $16.32b_1+ 21.81 b_2 + 8.10b_3 + 15.87b_4 $\\ \hline
$\delta_{2}^{2}(s_2)$ & $11.92b_1 + 12.32b_2 + 10.83b_3 + 14.36b_4 $ \\ \hline
$\delta_{2}^{3}(s_2)$ & $18.38b_1+ 1.97b_2 + 20b_3+ 0.43b_4 $ \\ \hline
$\delta_{2}^{4}(s_2)$ & $2.57b_1 + 21.27b_2 + 0.57b_3 + 30 b_4 $\\ \hline
\end{tabular}
\end{center}
\caption{Two-team zero-sum game with team $1$ and $2$  having private randomization over its strategies  }
\label{twoteamprivatrand3}
\end{table}
We illustrated the expected payoff matrix in table \ref{twoteamprivatrand3}.
If $\DM_1$ in Team $2$ do not play pure strategy, assume 
$b_1 = b_3 =0$, $b_2 =\frac{1}{4}$, and 
$b_4 = \frac{3}{4}$, then  $\DM_2$ of Team $2$ will play  strategy
 $\delta_{2}^{3}(s_2)$ to minimize the expected payoff. 
 Thus expected payoff $0.815$. Note that 
 $0.815 < \bar{V}_{A}$ but greater than pure strategy 
expected payoff (it is clear from table \ref{twoteamprivatrand2}) since expected payoff under pure strategy solution is $0.43$. Here we assume 
 if decision maker in Team $1$ having private randomization with probability vector
$a_3 = a_4 = \frac{1}{12}$, $a_1= \frac{5}{18}$ $a_2= \frac{10}{18}$.

\subsubsection{Role of common randomness independent of $\xi$}
Now consider common randomness independent of $\xi$ is provided to $\DM_1$ and $\DM_2$
of team $2$, i.e.
Team $2$ does joint randomization over its strategy 
then best for for team $2$ to put positive mass on 
strategies $(\delta_{1}^{4}(s_1),\delta_{2}^{3}(s_2))$
or  $(\delta_{1}^{3}(s_1),\delta_{2}^{4}(s_2))$. 
 Otherwise its expected payoff more than pure strategy 
 (it is clear from table \ref{twoteamexpectpay1}).
In discrete team vs team  zero-sum game with common randomness, do not admit Nash equilibrium 
solution. It also lead to improve in the expected payoff.

%% file: LQG-team-vs-team-zerosum.tex

\subsection{Example: LQG  team vs team zero-sum game}
\label{app:LQG-team-vs-team-zero-sumgame}

Now, we illustrate an example of LQG zero-sum team vs team game and show that common randomness 
independent of environment $\xi$ does not benefit.  We also demonstrate that common 
randomness dependent on $\xi$ benefit a team having extra randomness.

Consider two team LQG zero sum game, Team $1$  and Team $2$ consists of a decision maker 
and two decision makers, respectively. Let $\xi = [\mu_1, s_1, s_2 ]^{T} $ denote an 
environment or state of nature; it is random vector 
having probability distribution  $N(0, \Sigma)$, $\Sigma$ is covariance matrix.
Let $y_{i} = \eta_{i}(\xi)$ be the observations about $\xi$ available at decision maker $i$ of Team $1$, for $i = 1$; $z_j = \zeta_j(\xi) $ represents the observations about $\xi$ available at decision maker $j$ of Team $2$, for $j = 1, 2$. Mathematical simplicity, we assume $y_{1} = \eta_{1}(\xi) = \mu_1$,
$z_{j} = \zeta_{j}(\xi) = s_j$, $j = 1,2.$ 
In standard LQG two-team zero-sum game decision rule is defined as follows.
\begin{equation*}
\gamma_i: y_i \rightarrow u_i,
\end{equation*}
 $\gamma_i \in \Gamma_i$ and $u_i \in U_i$ for $i = 1$;
\begin{equation*}
\delta_j: z_j \rightarrow v_j,
\end{equation*}
$\delta_j \in \Delta_j $ and $v_j \in V_j $ for $j = 1,2$.

The optimal decision rule $\left(u_1^{*} = \gamma_{1}^{*}(y_1), v_{1}^{*} = \delta_{1}^{*}(z_1),v_{2}^{*} = \delta_{2}^{*}(z_2)  \right)$ such that
\begin{equation}
\mathcal{J}_{ZS, LQG}(u_1, v_1^{*}, v_2^{*}) \leq \mathcal{J}_{ZS,LQG}(u_1^{*}, v_1^{*}, v_2^{*}) \leq \mathcal{J}_{ZS,LQG}(u_1^{*}, v_1, v_2),
\label{eq:ZS_LQG_SaddlePt}
\end{equation}
for all $u_1 \in U_1$, $v_1 \in V_1$ and $v_2 \in V_2$; $\mathcal{J}_{ZS,LQG}(u_1, v_1, v_2) = \mathbb{E}_{\xi}[\kappa(u_1, v_1,v_2 , \xi) ]$.

The cost function:
\begin{eqnarray}
\nonumber
 \kappa(u_1, v_1, v_2 , \xi) & = & \kappa(\theta , \xi), \\
 & = &  \theta^{T}B \theta + 2 \theta^{T}S \xi,
\end{eqnarray} 
where $\theta = [u_1, v_1, v_2]^{T}$, $B = \left[ \begin{array}{ccc}
-1 & r_{11} & r_{12}  \\
r_{11} & 1 & q_{12} \\
r_{12} & q_{12} & 1  \end{array} \right], $
here $r_{11}$ and $r_{12}$ characterizes the coupling among teams, that is $r_{11}$ and $r_{12}$ is coupling of $\DM_{1}$ of Team $1$ with $\DM_1$ and $\DM_2$ of Team $2$ respectively. And $q_{12}$ denotes coupling among $\DM_1$ and $\DM_2$ of Team $2$. Moreover, we assume that Team $1$ seeks to maximize the expected payoff and Team $2$ seeks to minimize the expected payoff. It is required that the cost function $\mathbb{E}_{\xi}[\kappa(u_1, v_1,v_2 , \xi) ]$ to be concave in $u_1$ and convex in $v_1$ and $v_2$. Hence, we assume $1 - q_{12}^{2}  > 0$ and $S = \left[  \begin{array}{ccc}
1 & 0 & 0 \\
0 & -1 & 0 \\
0 & 0 & -1 \end{array} \right]$.
 
Two-team LQG zero-sum game admits a saddle point solution (for which we refer the
reader to  \cite[lemma 3.1, 3.2, theorem 3.1]{Ho74}\textit{•}), i.e.
\begin{equation}
\max_{u_1 \in U_1} \min_{(v_1, v_2) \in V} \mathbb{E}_{\xi}[\kappa(u_1, v_1,v_2 , \xi) ] =  \min_{(v_1, v_2) \in V} \max_{u_1 \in U_1} \mathbb{E}_{\xi}[\kappa(u_1, v_1,v_2 , \xi) ].
\label{eq: ZS_LQG_max_min}
\end{equation}
$ V = V_1 \times V_2$.
Since, in static LQG problem, decision variable are linear function of observations available at decision makers, $u_{1} = \gamma_{i}(y_{1}) = \alpha_{11} y_{1}$, $v_j =  \delta_j(z_j) = \alpha_{2j} z_j$, $j = 1,2.$

Re-writing relation of $\theta$ and observations $y_{1}$, $z_1$ $z_2$ more compactly, we have $\theta = A \tilde{y}$, where  $A = \left[ \begin{array}{ccc}
\alpha_{11} & 0 & 0  \\
0 & \alpha_{21} & 0  \\
 0 & 0 & \alpha_{22} \end{array} \right]$, and $\tilde{y} = [y_{1}, z_{1}, z_{2}]^{T}$. The expected cost function is 
 \begin{eqnarray}
\nonumber
\mathcal{J}_{ZS, LQG}(\alpha_{11}, \alpha_{21}, \alpha_{22}) & = &  \mathbb{E}_{\xi}[\tilde{y}^{T} A^{T} B A \tilde{y} + 2\tilde{y}^{T}A^{T}S \xi] \\
 & = & \Tr [ A^{T} B A \Sigma + 2 A^{T} S \Sigma].
\label{eq:ZS_LQG_EC}
\end{eqnarray}
Equality in \eqref{eq:ZS_LQG_EC} follows from $\tilde{y} = \xi$ and $\xi \sim N(0,\Sigma).$
Then we obtain from \eqref{eq: ZS_LQG_max_min}, 
\begin{equation}
\max_{\alpha_{11}} \min_{\alpha_{21}, \alpha_{22}} \mathcal{J}_{ZS, LQG}(\alpha_{11}, \alpha_{21}, \alpha_{22})  = \min_{\alpha_{21}, \alpha_{22}} \max_{\alpha_{11}} \mathcal{J}_{ZS, LQG}(\alpha_{11}, \alpha_{21}, \alpha_{22})
\end{equation}

An objective of zero-sum two team LQG game is to determine $(\alpha_{11}^{*}, \alpha_{21}^{*}, \alpha_{22}^{*})$ such that  

\begin{equation*}
\mathcal{J}_{ZS, LQG}(\alpha_{11}, \alpha_{21}^{*}, \alpha_{22}^{*}) \leq \mathcal{J}_{ZS,LQG}(\alpha_{11}^{*}, \alpha_{21}^{*}, \alpha_{22}^{*}) \leq \mathcal{J}_{ZS,LQG}(\alpha_{11}^{*}, \alpha_{21}, \alpha_{22}),
\end{equation*}
 will be satisfied for $\alpha_{11}, \alpha_{21}, \alpha_{22} \in \mathbb{R}$.
\subsubsection{Discussion on matrix $B$ } 
In matrix $B$, we have coupling parameter $r_{11}$, $r_{12}$ and $q_{12}$. If $r_{11} = r_{12} = q_{12} = 0$, there is no coupling among Team $1$ and $2$, as well as among decision makers of Team $2$. This is not at all interesting. If $r_{11} = r_{12}  = 0 $, then there is no coupling among team $1$ and $2$. Problem becomes  team decision problem. Hence we suppose $r_{11}, r_{12}, q_{12} \neq 0.$

Next, we analyze the role of common randomness in LQG two-team zero-sum game. We describe two cases as follows.
\begin{itemize}
\item \textbf{Case I:} Common randomness independent of $\xi$.
\item \textbf{Case II:} Common randomness dependent on $\xi$.
\end{itemize}
\subsubsection{Common randomness independent of $\xi$}
\label{app:common-info-indep-LQG}
\begin{proposition}
In LQG two-team zero-sum stochastic game, common randomness independent of $\xi$ do not benifit the team.
\end{proposition}
\begin{proof}
Consider a coordinator provides  common randomness which is independent
of environment $\xi$  to the decision makers of teams. 
For mathematical simplicity, we assume common randomness is available at one of team, say  Team $2$.
The common randomness provided to decision maker $\DM_1$ and $\DM_2$ of team $2$ is represented as  $\omega$, and also $\omega \amalg \xi $. The decision rule of a decision maker of Team $1$ is
\begin{equation*}
\gamma_{1} : y_{1} \rightarrow u_{1},
\end{equation*}
and decision rule of Team $2$ decision makers are 
\begin{equation*}
\delta_{j} : z_{j} \times \omega \rightarrow v_{j},
\end{equation*}
$j = 1,2$. Actions of decision makers are
\begin{equation*}
u_{1} = \gamma_{1}(y_1) = \alpha_{11} y_{1},
\end{equation*}
\begin{equation*}
v_{j} = \delta_{j}(z_j, \omega) =  \alpha_{2j} z_{j} + \beta_{2j} \omega,
\end{equation*}
for $j = 1,2$. Rewriting above expression, we obtain
\begin{equation*}
\theta = A \tilde{y} + \beta \omega,
\end{equation*}
here, $\theta = 
 \left[ \begin{array}{c}
 u_{1}  \\
 v_{1}  \\
 v_{2} \end{array} \right] $,
 $A = \left[ \begin{array}{ccc}
\alpha_{11} & 0 & 0  \\
0 & \alpha_{21} & 0  \\
 0 & 0 & \alpha_{22} \end{array} \right] $, 
 $\tilde{y} = 
  \left[ \begin{array}{c}
 y_{1}  \\
 z_{1}  \\
 z_{2} \end{array} \right]$, $\beta = \left[ \begin{array}{c}
 0  \\
 \beta_{21}  \\
 \beta_{22} \end{array} \right]$.
 
The expected payoff of LQG two team zero-sum game with common randomness is
\begin{eqnarray}
\nonumber
\mathcal{J}_{ZS, CR, LQG} (\alpha_{11}, \alpha_{21}, \alpha_{22}, \beta)& = & \mathbb{E}_{\xi}[\tilde{y}^{T} A^{T} B A \tilde{y} + 2 \tilde{y}^{T}A^{T}S \xi + 2 \tilde{y}^{T} A^{T} B \beta \omega + \omega^{T} \beta^{T} B \beta \omega + 2 \omega^{T} \beta^{T} S \xi], \\
 & = &  \Tr [ A^{T} B A \Sigma + 2 A^{T} S \Sigma + \beta^{T} B \beta \Sigma_{2}]. 
 \label{eq:ZS_CR_LQG1}
  \label{eq:ZS_CR_LQG2}
\end{eqnarray} 
Equality in \eqref{eq:ZS_CR_LQG1} because $\omega \amalg \xi$, $\omega \sim N(0, \Sigma_2).$ 
\begin{eqnarray*}
\nonumber
\max_{\alpha_{11}} \min_{\alpha_{21}, \alpha_{22}, \beta_{21}, \beta_{22}} \mathcal{J}_{ZS, CR, LQG} (\alpha_{11}, \alpha_{21}, \alpha_{22}, \beta_{21}, \beta_{22})  =  \max_{\alpha_{11}} \min_{\alpha_{21}, \alpha_{22}, \beta_{21}, \beta_{22}}\Tr [ A^{T} B A \Sigma + 2 A^{T} S \Sigma + \beta^{T} B \beta \Sigma_{2} ], \\
\nonumber
 =  \max_{\alpha_{11}} \min_{\alpha_{21}, \alpha_{22}}\Tr [ A^{T} B A \Sigma + 2 A^{T} S \Sigma] + \min_{\beta_{21}, \beta_{22}} \Tr[\beta^{T} B \beta \Sigma_{2}].
\end{eqnarray*}
 
 Clearly, from  above expression, minimization of  $ \Tr[\beta^{T} B \beta \Sigma_{2}]$  attained at $\beta $ equals to zero, i.e.$\beta_{11} = 0$, $\beta_{21} = 0$, $\beta_{22} = 0$ for given $B$ and $\Sigma_2 > 0.$
\begin{eqnarray*}
\nonumber
\max_{\alpha_{11}} \min_{\alpha_{21}, \alpha_{22}, \beta_{21}, \beta_{22}} \mathcal{J}_{ZS, CR, LQG} (\alpha_{11}, \alpha_{21}, \alpha_{22}, \beta_{21}, \beta_{22}) 
 & = &  \max_{\alpha_{11}} \min_{\alpha_{21}, \alpha_{22}}\Tr [ A^{T} B A \Sigma + 2 A^{T} S \Sigma]\\
& = & \max_{\alpha_{11}} \min_{\alpha_{21}, \alpha_{22}} \mathcal{J}_{ZS,  LQG} (\alpha_{11}, \alpha_{21}, \alpha_{22}) \\
& = &  \min_{\alpha_{21}, \alpha_{22}, \beta_{21}, \beta_{22}} \max_{\alpha_{11}} \mathcal{J}_{ZS, CR, LQG} (\alpha_{11}, \alpha_{21}, \alpha_{22}, \beta_{21}, \beta_{22}) 
\end{eqnarray*}
  Hence we conclude that common randomness independent of $\xi$ do not benefit the team having common randomness.
\end{proof}
 
\subsubsection{Common randomness dependent on $\xi$}
\label{app:common-randomness-depend-envi-LQG}
Suppose  the common randomness available at decision makers of Team $2$ of two-team LQG zero-sum game; it is denoted as $\omega$.
The decision rule of a decision maker in Team $1$ is
\begin{equation*}
\gamma_{1} : y_{1} \rightarrow u_{1},
\end{equation*}
and decision rule of Team $2$ decision makers are 
\begin{equation*}
\delta_{j} : z_{j} \times \omega \rightarrow v_{j},
\end{equation*}
$j = 1,2$. Actions of decision makers are
\begin{equation*}
u_{1} = \gamma_{1}(y_1) = \alpha_{11} y_{1},
\end{equation*}
\begin{equation*}
v_{j} = \delta_{j}(z_j, \omega) =  \alpha_{2j} z_{j} + \beta_{2j} \omega,
\end{equation*}
for $j = 1,2$. We have6
\begin{equation*}
\theta = A \tilde{y} + \beta \omega,
\end{equation*}
here, $\theta = 
 \left[ \begin{array}{c}
 u_{1}  \\
 v_{1}  \\
 v_{2} \end{array} \right] $,
 $A = \left[ \begin{array}{ccc}
\alpha_{11} & 0 & 0  \\
0 & \alpha_{21} & 0  \\
 0 & 0 & \alpha_{22} \end{array} \right] $, 
 $\tilde{y} = 
  \left[ \begin{array}{c}
 y_{1}  \\
 z_{1}  \\
 z_{2} \end{array} \right]$, $\beta = \left[ \begin{array}{c}
 0  \\
 \beta_{21}  \\
 \beta_{22} \end{array} \right]$.

 Moreover 
it is assume that the common randomness is dependent on an environment $\xi$. Hence $\omega$ is function of $\xi$, that is $\omega = f(\xi)$;
$f(\cdot)$ is measurable function. Let $f$ be the linear function, then
\begin{eqnarray*}
\omega &=& f(\xi) 
        =  \phi_{11} \mu_1 + \phi_{21} s_1 + \phi_{22} s_2 \\
       & = & \Phi^{T} \tilde{y} 
        =  \Phi^{T} \xi.
\end{eqnarray*}
Where $\Phi = [\phi_{11}, \phi_{21}, \phi_{22}]^{T}$, $\tilde{y} = \xi$ and $\xi \sim N(0, \Sigma)$.
The expected cost functional is
\begin{eqnarray}
\nonumber 
\mathcal{J}_{ZS, CR, LQG} (\alpha_{11}, \alpha_{21}, \alpha_{22}, \beta_{21}, \beta_{22})& = & \mathbb{E}_{\xi}[\tilde{y}^{T} A^{T} B A \tilde{y} + 2 \tilde{y}^{T}A^{T}S \xi + 2 \tilde{y}^{T} A^{T} B \beta \omega + \omega^{T} \beta^{T} B \beta \omega + 2 \omega^{T} \beta^{T} S \xi], \\
 & = &  \Tr [ A^{T} B A \Sigma + 2 A^{T} S \Sigma + 2 A^{T} B \tilde{\beta} \Sigma + \tilde{\beta}^{T} B \tilde{\beta} \Sigma + 2 \tilde{\beta}^{T} S \Sigma]. 
 \label{eq:ZS_CR_LQG1}
  \label{eq:ZS_CR_LQGd}
\end{eqnarray} 

In \eqref{eq:ZS_CR_LQGd}, $\tilde{\beta} = \beta \Phi^{T}$. Goal is to  find $(\alpha_{11}^{*}, \alpha_{21}^{*}, \alpha_{22}^{*}, \beta_{21}^{*}, \beta_{22}^{*})$ such that 
\begin{equation*}
\mathcal{J}_{ZS, CR, LQG} (\alpha_{11}, \alpha_{21}^{*}, \alpha_{22}^{*}, \beta_{21}^{*}, \beta_{22}^{*})
\leq 
\mathcal{J}_{ZS, CR, LQG} (\alpha_{11}^{*}, \alpha_{21}^{*}, \alpha_{22}^{*}, \beta_{21}^{*}, \beta_{22}^{*})
\leq 
\mathcal{J}_{ZS, CR, LQG} (\alpha_{11}^{*}, \alpha_{21}, \alpha_{22}, \beta_{21}, \beta_{22})
\end{equation*}
for $\alpha_{11},\alpha_{21},\alpha_{22}, \beta_{21}, \beta_{22} \in \mathbb{R}$.

Source of information (source of common randomness) can act as a \textit{mole} or \textit{consultant} depending on type of information it provides. If source of information is a \textit{mole} then $\omega = \phi_{11} \mu_1$. It implies $\phi_{21} = 0$, and $\phi_{22} = 0$. If source of information is \textit{consultant}, then  $\omega = \phi_{21}s_1 + \phi_{22}s_2$. We will investigate two different cases based on source of information and types of information it provides.

a) Suppose the source of information is a mole or spy and it provide information (common randomness) $\omega = \phi_{11} \mu_1$. Let $\mathcal{J}^{a,*}_{ZS, CR, LQG} $ denote the saddle point solution of LQG two-team zero-sum game with common randomness when source of common randomness to Team $2$ decision makers  is spy.

b) Let $\mathcal{J}^{b,*}_{ZS, CR, LQG} $ represents the saddle point solution of LQG two-team zero-sum game with common randomness when source of common randomness to Team $2$ decision makers  is consultant and $\omega = \phi_{21}s_1 + \phi_{22}s_2$.

 Intuitively, we expect to have following inequalities.
\begin{equation}
 \mathcal{J}^{a,*}_{ZS, CR, LQG} \leq \mathcal{J}^{*}_{ZS, LQG}.
\label{eq:ZS_CR_LQG_Mole}
\end{equation}
\begin{equation}
 \mathcal{J}^{b,*}_{ZS, CR, LQG} \leq \mathcal{J}^{*}_{ZS, LQG}.
\label{eq:ZS_CR_LQG_Consult}
\end{equation}
Note $\mathcal{J}^{*}_{ZS, LQG}$ is saddle point solution of LQG two-team zero-sum game with no common randomness.

From \eqref{eq:ZS_CR_LQGd}, analytically, it is difficult to prove the inequalities in \eqref{eq:ZS_CR_LQG_Mole}, \eqref{eq:ZS_CR_LQG_Consult}. Hence we conjecture result in \eqref{eq:ZS_CR_LQG_Mole}, \eqref{eq:ZS_CR_LQG_Consult}.
Now we present numerical results and show that above inequalities  are true.
Let $\Sigma = \left( \begin{array}{ccc}
\sigma_{\mu_1}^{2} & \sigma_{\mu_1,s_1}^{2} & \sigma_{\mu_1,s_2}^{2} \\
\sigma_{\mu_1,s_1}^{2} & \sigma_{s_1}^{2} & \sigma_{s_1,s_2}^{2} \\
\sigma_{\mu_1,s_2}^{2} & \sigma_{s_1,s_2}^{2} & \sigma_{s_2}^{2} \end{array} \right)$,
Since $\omega$ is scalar, we have $\Sigma_2 =  \sigma_{\omega}^{2}$.
Team cost functional is 
\begin{eqnarray}
\nonumber
\mathcal{J}( \alpha_{11}, \alpha_{21}, \alpha_{22}, \beta_{21}, \beta_{22})= -\alpha_{11}^{2}\sigma_{\mu_1}^{2} +
\alpha_{21}^{2} \sigma_{s_1}^{2}+ \alpha_{22}^{2}\sigma_{s_2}^{2} 
+ 2 r_{11}\alpha_{11}\alpha_{21}  \sigma_{\mu_1,s_1}^{2} 
+ 2 r_{12}\alpha_{11}\alpha_{22} \sigma_{\mu_1,s_2}^{2} \\ \nonumber
+ 2 q_{12} \alpha_{21} \alpha_{22} \sigma_{s_1,s_2}^{2}  
+ 2(r_{11}\alpha_{11}\beta_{21}+ r_{12}\alpha_{11}\beta_{22} ) \sigma_{\mu_1, w}^{2} + 
2 (\alpha_{21}\beta_{21}+ q_{12}\alpha_{21}\beta_{22} ) \sigma_{s_1, w}^{2} \\ \nonumber
+ 2(q_{12}\alpha_{22}\beta_{21}+ \alpha_{22}\beta_{22})  \sigma_{s_2, w}^{2} 
+ (\beta_{21}^{2}+ 2q_{12}\beta_{21}\beta_{22} + \beta_{22}^{2})\sigma_{w}^{2} 
+ 2 \alpha_{11} \sigma_{\mu_1}^2 \\
- 2 \alpha_{21} \sigma_{s_1}^{2} - 2 \alpha_{22} \sigma_{s_2}^{2}
- 2 \beta_{21} \sigma_{s_1, w}^{2} - 2 \beta_{22} \sigma_{s_2, w}^{2}.
\label{eq:ZS_LQG_NEX}
\end{eqnarray}
We know that LQG two-team zero-sum game has saddle point solution, that is 
\begin{equation}
\max_{\alpha_{11}} \min_{\alpha_{21}, \alpha_{22}, \beta_{21}, \beta_{22}}\mathcal{J}_{ZS,CR,LQG}( \alpha_{11}, \alpha_{21}, \alpha_{22}, \beta_{21}, \beta_{22}) 
= \min_{\alpha_{21}, \alpha_{22}, \beta_{21}, \beta_{22}} \max_{\alpha_{11}} \mathcal{J}_{ZS,CR,LQG}( \alpha_{11}, \alpha_{21}, \alpha_{22}, \beta_{21}, \beta_{22}).
\end{equation}
To evaluate $\max_{\alpha_{11}} \min_{\alpha_{21}, \alpha_{22}, \beta_{21}, \beta_{22}}\mathcal{J}_{ZS,CR,LQG}( \alpha_{11}, \alpha_{21}, \alpha_{22}, \beta_{21}, \beta_{22}) $, we differentiate \eqref{eq:ZS_LQG_NEX} with respect to $\alpha_{11}, \alpha_{21}, \alpha_{22}, \beta_{21}, \beta_{22}$ and equate to $0$. We obtain linear systems of equations as follows.
\begin{eqnarray*}
\left[  \begin{array}{ccccc}
-\sigma_{\mu_1}^{2} & r_{11}\sigma_{\mu_1,s_1}^{2} & r_{12}\sigma_{\mu_1,s_2}^{2} & r_{11} \sigma_{\mu_1, w}^{2} & r_{12}\sigma_{\mu_1, w}^{2}\\
r_{11}\sigma_{\mu_1,s_1}^{2} & \sigma_{s_1}^{2} & q_{12}\sigma_{s_{1}s_2}^{2} & \sigma_{s_1, w}^{2} & q_{12}\sigma_{s_1,w}^{2}\\
r_{12}\sigma_{\mu_1, s_2}^{2} & q_{12}\sigma_{s_1, s_2}^{2} & \sigma_{s_2}^{2} & q_{12}\sigma_{s_2, w}^{2} & \sigma_{s_2, w}^{2} \\
 r_{11} \sigma_{\mu_1, w}^{2}& \sigma_{s_1, w}^{2} & q_{12}\sigma_{s_2,w}^{2} & \sigma_{w}^{2} & q_{12}\sigma_{w}^{2} \\
r_{12}\sigma_{\mu_1, w}^{2} & q_{12}\sigma_{s_1, w}^{2} & \sigma_{s_2, w}^{2} & q_{12}\sigma_{w}^{2} & \sigma_{w}^{2}
 \end{array} \right]
\left[  \begin{array}{c}
\alpha_{11} \\
\alpha_{21} \\
\alpha_{22} \\
\beta_{21} \\
\beta_{22}
\end{array} \right]
= 
\left[  \begin{array}{c}
-\sigma_{\mu_1}^{2} \\
\sigma_{s_1}^{2} \\
\sigma_{s_2}^{2} \\
\sigma_{s_1, w}^{2} \\
\sigma_{s_2, w}^{2}
\end{array} \right].
\end{eqnarray*}
Numerically, we compare our result for different values matrix $B$.

1)$B = \left[ \begin{array}{ccc}
-1 & \frac{1}{4} & \frac{1}{4}  \\
\frac{1}{4} & 1 & \frac{1}{2}  \\
 \frac{1}{4} & \frac{1}{2} & 1 \end{array} \right]$,
2)$B = \left[ \begin{array}{ccc}
-1 & \frac{1}{4} & \frac{1}{2}  \\
\frac{1}{4} & 1 & \frac{1}{2}  \\
 \frac{1}{2} & \frac{1}{2} & 1 \end{array} \right]$,

We assume
$\Sigma =  
\left[  \begin{array}{ccc}
2 & \frac{1}{4} & \frac{1}{4} \\
\frac{1}{4} & 1 & \frac{1}{2} \\
\frac{1}{4} & \frac{1}{2} & 1 \end{array} \right] $ for all numerical results.  
 
\textbf{a)} When source of information is a \textit{mole} and $\omega = \phi_{11} \mu_1$, we have
$\mathbb{E}[\omega] = 0$,
\[ \mathbb{E}[\omega^{2}] = \sigma_{\omega}^{2} = \phi_{11}^{2} \sigma_{\mu_1}^{2}.\]
\[ \sigma_{\mu_1, \omega}^{2} = \phi_{11} \sigma_{\mu_1}^{2} \]
\[ \sigma_{s_1, \omega}^{2} = \phi_{11} \mathbb{E} [\mu_1 s_1] = \phi_{11} \sigma_{\mu_1,s_1}^{2}. \]
\[ \sigma_{s_2, \omega}^{2} = \phi_{11} \mathbb{E}[\mu_1 s_2] = \phi_{11} \sigma_{\mu_1,s_2}^{2}. \]

Case 1) $B = \left[  \begin{array}{ccc}
-1 & \frac{1}{4} & \frac{1}{4} \\
\frac{1}{4} & 1 & \frac{1}{2} \\
\frac{1}{4} & \frac{1}{2} & 1 \end{array} \right]$

%
%
%
After solving linear systems of equation, we have

$\alpha_{11}^{*} = 0.9615$, $\alpha_{21}^{*} = 0.8052$, $\alpha_{22}^{*} = 0.8052$, $\beta_{21}^{*} = -0.7103$, $\beta_{22}^{*} = -0.7103.$
Team cost functional is 
\begin{eqnarray*}
 \mathcal{J}_{ZS,CR,LQG}^{a,*} =  \mathcal{J}_{ZS,CR,LQG}^{a} (\alpha_{11}^{*},\alpha_{21}^{*},\alpha_{22}^{*}, \beta_{21}^{*}, \beta_{22}^{*}) & = &
\max_{\alpha_{11}} \min_{\alpha_{21}, \alpha_{22}, \beta_{21}, \beta_{22}}\mathcal{J}( \alpha_{11}, \alpha_{21}, \alpha_{22}, \beta_{21}, \beta_{22}) \\
& = &
0.4012.
\end{eqnarray*}

Case 2) $B = \left[ \begin{array}{ccc}
-1 & \frac{1}{4} & \frac{1}{2}  \\
\frac{1}{4} & 1 & \frac{1}{2}  \\
 \frac{1}{2} & \frac{1}{2} & 1 \end{array} \right]$.
%

Solving linear systems of equations we obtain
$\alpha_{11}^{*} = 0.8500$, $\alpha_{21}^{*} =0.8052$, $\alpha_{22}^{*} = 0.8052$ $\beta_{21}^{*} =-0.0693$, $\beta_{22}^{*} = -1.7693$. Evaluating team cost functional
\begin{eqnarray*}
 \mathcal{J}_{ZS,CR,LQG}^{a,*} =  0.2037.
\end{eqnarray*}

\begin{table}

\begin{center}
\begin{tabular}{ |l|l|l| }
\hline
$(r_{11}, r_{12}, q_{12})$ & $(\phi_{11},\phi_{21}, \phi_{22})$   & $\mathcal{J}_{ZS,CR,LQG}^{a,*} $  \\ \hline
$(\frac{1}{4}, \frac{1}{4}, \frac{1}{2})$ & $(\frac{1}{2}, 0, 0)$ & $0.4012$  \\ \hline
$(\frac{1}{4}, \frac{1}{2}, \frac{1}{2})$ & $(\frac{1}{2}, 0, 0 )$ & $0.2037 $  \\
  \hline
\end{tabular}
\caption{With Randomization: Comparison of $\mathcal{J}_{ZS,CR,LQG}^{a,*} $  for different values  of $r_{11}$, $r_{12}$, $q_{12}$. }
\label{table:ZS_CR_LQG_a}
\end{center}

\end{table}
\begin{table}

\begin{center}
\begin{tabular}{ |l|l|l| }
\hline
$(r_{11}, r_{12}, q_{12})$ & $(\phi_{11},\phi_{21}, \phi_{22})$   & $\mathcal{J}_{ZS,CR,LQG}^{b,*} $  \\ \hline
$(\frac{1}{4}, \frac{1}{4}, \frac{1}{2})$ & $(0, \frac{1}{2}, \frac{1}{2})$ & $0.1616$  \\ \hline
$(\frac{1}{4}, \frac{1}{2}, \frac{1}{2})$ & $(0, \frac{1}{2}, \frac{1}{2} )$ & $0.2435 $  \\
  \hline
\end{tabular}
\caption{With Randomization: Comparison of $\mathcal{J}_{ZS,CR,LQG}^{b,*} $  for different values  of $r_{11}$, $r_{12}$, $q_{12}$. }
\label{table:ZS_CR_LQG_b}
\end{center}

\end{table}
\begin{table}

\begin{center}
\begin{tabular}{ |l|l|l| }
\hline
$(r_{11}, r_{12}, q_{12})$ & $(\phi_{11},\phi_{21}, \phi_{22})$   & $\mathcal{J}_{ZS,LQG}^{*} $  \\ \hline
$(\frac{1}{4}, \frac{1}{4}, \frac{1}{2})$ & $(0, 0, 0)$ & $0.598$  \\ \hline
$(\frac{1}{4}, \frac{1}{2}, \frac{1}{2})$ & $(0,0, 0 )$ & $1.8991$  \\
  \hline
\end{tabular}
\caption{Without Randomization: Comparison of $\mathcal{J}_{ZS,LQG}^{*} $  for different values  of $r_{11}$, $r_{12}$, $q_{12}$. }
\label{table:ZS_LQG}
\end{center}

\end{table}

\textbf{b)} When a consultant provides an information, $\omega = \phi_{21} s_1 + \phi_{22} s_2 $.
Note that $\mathbb{E}[w] = 0$,

\[\sigma_{w}^{2} = \mathbb{E}[w^{2}] = \phi_{21}^{2} \sigma_{s_1}^{2} + \phi_{22}^{2} \sigma_{s_2}^{2} + 2 \phi_{21} \phi_{22} \sigma_{s_1, s_2}^{2}.\]
\[\sigma_{\mu_1, w}^{2} = \mathbb{E}[\mu_{1} w] = \phi_{21} \sigma_{\mu_1,s_1}^{2} + \phi_{22} \sigma_{\mu_1,s_2}^{2}.\]
\[\sigma_{s_1, w}^{2} = \mathbb{E}[s_{1} w] = \phi_{21} \sigma_{s_1}^{2} + \phi_{22} \sigma_{s_1,s_2}^{2}.\]
\[\sigma_{s_2, w}^{2} = \mathbb{E}[s_{2} w] = \phi_{21} \sigma_{s_1,s_2}^{2} + \phi_{22} \sigma_{s_2}^{2}.\]
We suppose $\phi_{21} =  \frac{1}{2}$, $\phi_{22} = \frac{1}{2}$.

Case 1) $B = \left[  \begin{array}{ccc}
-1 & \frac{1}{4} & \frac{1}{4} \\
\frac{1}{4} & 1 & \frac{1}{2} \\
\frac{1}{4} & \frac{1}{2} & 1 \end{array} \right]$
Solving linear system of eqaution we have
$\alpha_{11}^{*} = 1.0381$, $\alpha_{21}^{*} = 2$, $\alpha_{22}^{*} = 2$, $\beta_{21}^{*} = -1.391$, $\beta_{22}^{*} =  -1.391$ and team optimal cost  $ \mathcal{J}_{ZS,CR,LQG}^{b,*}  = 0.1616$.

Case 2)$B = \left[ \begin{array}{ccc}
-1 & \frac{1}{4} & \frac{1}{2}  \\
\frac{1}{4} & 1 & \frac{1}{2}  \\
 \frac{1}{2} & \frac{1}{2} & 1 \end{array} \right]$.
 Then $\alpha_{11}^{*} =  1.0515$, $\alpha_{21}^{*} = 2$, $\alpha_{22}^{*} = 2$, $\beta_{21}^{*} = -1.3333$, $\beta_{22}^{*} =  -1.5086$ and team optimal cost  $ \mathcal{J}_{ZS,CR,LQG}^{b,*}  = 0.2435$.

From table \ref{table:ZS_CR_LQG_a}, \ref{table:ZS_CR_LQG_b}, \ref{table:ZS_LQG}, it clear that  inequalities in \eqref{eq:ZS_CR_LQG_Mole},\eqref{eq:ZS_CR_LQG_Consult} satisfy numerically.
Observe that common randomness dependent on $\xi$ provided by either a \textit{mole} or \textit{consultant} benefits the team vs team zero-sum game.